\documentclass{article}

\usepackage[affil-it]{authblk}
\usepackage[usenames,dvipsnames]{xcolor}
\usepackage{amsfonts}
\usepackage{amsmath,amsthm,amssymb,dsfont}
\usepackage{enumerate}
\usepackage[english]{babel}
\usepackage{graphicx}	
\usepackage[caption=false]{subfig}
\usepackage[left=2.6cm,right=2.6cm,top=2cm,bottom=2.1cm]{geometry}
\usepackage{url}
\usepackage{todonotes}
\usepackage{bbm}
\usepackage{booktabs}
\usepackage{stmaryrd}
\usepackage{tikz}
\usetikzlibrary{chains}
\usetikzlibrary{fit}
\usepackage{pgflibraryarrows}		%optional
\usepackage{pgflibrarysnakes}		%optional
\usepackage{multirow}
\usepackage{makecell}
\usepackage{epsfig}
\usetikzlibrary{shapes.symbols,patterns} % for source symbols
\usepackage{pgfplots}
\usepackage{bm}
\pgfplotsset{compat=newest}

\usepackage{hyperref}
\hypersetup{colorlinks=true,citecolor=blue,linkcolor=blue,filecolor=blue,urlcolor=blue,breaklinks=true}

\usepackage{nicefrac}
\usepackage{mathtools}
\usepackage{mathrsfs}
\usepackage{IEEEtrantools}

\theoremstyle{plain}
\newtheorem{theorem}{Theorem}[section]
\newtheorem{lemma}[theorem]{Lemma}

\newtheorem{proposition}[theorem]{Proposition}

\theoremstyle{definition}

\newtheorem{remark}[theorem]{Remark}

\newtheorem*{thm}{Theorem}

\newcommand\blfootnote[1]{%
	\begingroup
	\renewcommand\thefootnote{}\footnote{#1}%
	\addtocounter{footnote}{-1}%
	\endgroup
}

\newcommand*{\cA}{\mathcal{A}}
\newcommand*{\cB}{\mathcal{B}}

\newcommand*{\cD}{\mathcal{D}}
\newcommand*{\cE}{\mathcal{E}}

\newcommand*{\cH}{\mathcal{H}}

\newcommand*{\cI}{\mathcal{I}}

\newcommand*{\cN}{\mathcal{N}}

\newcommand*{\cJ}{\mathcal{J}}

\newcommand*{\cS}{\mathcal{S}}

\newcommand*{\cW}{\mathcal{W}}

\newcommand*{\SDP}{\mathrm{SDP}}

\newcommand*{\CC}{\mathbb{C}}

\newcommand*{\NN}{\mathbb{N}}

\newcommand*{\GL}{\mathrm{GL}}

\newcommand*{\End}{\mathrm{End}}

\newcommand*{\Par}{\mathrm{Par}}
\newcommand*{\id}{I}
\newcommand*{\poly}{\mathrm{poly}}

\newcommand*{\tr}{\mathrm{tr}\,}
\newcommand*{\ket}[1]{| #1 \rangle}
\newcommand*{\bra}[1]{\langle #1 |}

\newcommand{\ketbra}[2]{|#1\rangle\!\langle #2|}

\newcommand*{\bigO}{\mathrm{O}}

\newcommand*{\Pos}{\mathscr{P}}
\newcommand*{\Lin}{\mathscr{L}}

\newcommand*{\CP}{\mathrm{CP}}

\newcommand*{\mS}{\mathfrak{S}}
\newcommand*{\<}{\langle}
\renewcommand*{\>}{\rangle}

\newcommand {\br} [1] {\ensuremath{ \left( #1 \right) }}

\newcommand {\cbr} [1] {\ensuremath{ \left\lbrace #1 \right\rbrace }}

\usepackage{scalerel,stackengine}
\newcommand\reallywidehat[1]{%
	\savestack{\tmpbox}{\stretchto{%
			\scaleto{%
				\scalerel*[\widthof{\ensuremath{#1}}]{\kern-.6pt\bigwedge\kern-.6pt}%
				{\rule[-\textheight/2]{1ex}{\textheight}}%WIDTH-LIMITED BIG WEDGE
			}{\textheight}% 
		}{0.5ex}}%
	\stackon[1pt]{#1}{\tmpbox}%
}

\newcommand{\suppress}[1]{}

\begin{document}

%\title{{\LARGE An efficient bound on quantum channel fidelity exploiting symmetry}}
\title{{\LARGE An efficient parameterized algorithm for computing quantum channel fidelity via symmetries exploitation}}

%\author{Yeow Meng Chee, 
%	Hoang Ta, and Van Khu Vu
%}
%%\author{}
%\affil{Department of Industrial Systems Engineering and Management,\\ National University of Singapore, Singapore}

\author[1]{Yeow Meng Chee}
\author[2]{Hoang Ta}
\author[3]{Van Khu Vu}

\affil[1]{\small{Singapore University of Technology and Design, Singapore}}
\affil[2]{\small{Hanoi University of Science and Technology, Vietnam}}
\affil[3]{\small{VinUniversity, Vietnam}}

\date{}

\maketitle

%\tableofcontents

\begin{abstract}
	Determining the optimal fidelity for the transmission of quantum information over noisy quantum channels is one of the central problems in quantum information theory. Recently, [Berta-Borderi-Fawzi-Scholz, Mathematical Programming, 2021] introduced an asymptotically converging semidefinite programming hierarchy of outer bounds for this quantity. However, the size of the semidefinite programs (SDPs) grows exponentially with respect to the level of the hierarchy, thus making their computation unscalable. In this work, by exploiting the symmetries in the SDP, we show that, for a fixed output dimension of the quantum channel, we can compute the SDP in time polynomial with respect to the level of the hierarchy and input dimension. As a direct consequence of our result, the optimal fidelity can be approximated with an accuracy of $\epsilon$ in $\poly(1/\epsilon, \text{input dimension})$ time, compared to the $\exp(1/\epsilon, \text{input dimension})$ running time required for direct computation.
\end{abstract}
%This results in a speed-up for computing SDP hierarchy in BBFS.
%%%%%%%%%%%%%%%%%%%%%%%%%%%%%%%%%%%%%%%%%%%%%%%%%%%%%%
\maketitle
\blfootnote{Part of this work was completed while all three authors were with the Department of Industrial Systems Engineering and Management, National University of Singapore, Singapore.}

\section{Introduction}
\label{sec:introduction}

One of the central problems in information theory is determining the minimum error probability that arises during data transmission through a noisy channel or when storing it on an unreliable storage medium. This task can be formulated as an optimization problem that maximizes the probability of determining the sent messages exactly over all valid encoders and decoders. This problem was initially investigated by Shannon~\cite{shannon1948mathematical}. There, the author showed that when considering large numbers of independent copies of the same channel, this quantity can be characterized by a simple expression known as the channel capacity of the channel.

\par 
From an algorithmic perspective, this problem was studied in~\cite{barman2017algorithmic}, which focused on investigating the problem of determining the optimal encoder and decoder that maximize the success probability over a noisy channel in the non-asymptotic regime. In this setting, given $N_{X \to Y}$ that is a noisy channel from $X$ to $Y$ and $M \in \NN$, let $\mathrm{p}(N,M)$ be the maximum success probability for transmitting a uniform $M$-dimensional message through the channel $N_{X\to Y}$. The study in~\cite{barman2017algorithmic} showed that the computation of $\mathrm{p}(N,M)$ can be formulated as an optimization problem involving a submodular function. Building upon this insight, a simple and efficient greedy algorithm was introduced in~\cite{barman2017algorithmic} to find a code that achieves a $(1-e^{-1})$-approximation of $\mathrm{p}(N,M)$. Furthermore, the study in~\cite{barman2017algorithmic} provided an efficiently computable linear programming relaxation denoted by $\mathrm{LP}(N,M)$ (also referred to as meta converse~\cite{hayashi2009information,polyanskiy2010channel}). This relaxation provides upper bounds on $\mathrm{p}(N,M)$, and it was proven that $\mathrm{p}(N,M) \leq \mathrm{LP}(N,M) \leq (1-e^{-1})^{-1} \mathrm{p}(N,M)$. However, it was highlighted in~\cite{barman2017algorithmic} that the problem of obtaining an approximate solution for $\mathrm{p}(N,M)$ with a better constant factor than $1-e^{-1}$ is classified as $\mathrm{NP}$-Hard. Further research on this problem on variant models has been conducted in~\cite{berta2016quantum,fawzi2019approximation,barman2020tight,fawzi2023broadcast}.

\par

The problem in the analog quantum setting involves determining the quantum channel fidelity (or short channel fidelity) $\mathrm{F}(\cN,M)$ for transmitting one part of a maximally entangled state with dimension $M$ through a noisy quantum channel $\cN_{A \to B}$ from $A$ to $B$, where $A$ and $B$ are finite dimensional
complex Hilbert spaces. Similar to the classical case, this problem can be formulated as a bilinear optimization problem with matrix-valued variables. In order to approximate $\mathrm{F}(\cN,M)$, an efficiently computable semidefinite programming relaxation was proposed in~\cite{leung2015power}. However, unlike the classical case, the gap between this relaxation and $\mathrm{F}(\cN,M)$ is not well understood in the quantum setting. In other directions, numerical methods based on iterative seesaw techniques, such as those proposed in~\cite{reimpell2005iterative,fletcher2007optimum,taghavi2010channel,johnson2017qvector}, have been developed to provide lower bounds for $\mathrm{F}(\cN,M)$. These methods are computationally tractable semidefinite programs that often converge in practice. In different settings, several works in~\cite{tomamichel2016quantum,wang2016semidefinite,wang2018semidefinite,kaur2019extendibility} have focused on the problem of determining the size of a maximally entangled state that can be transmitted through a noisy quantum channel while maintaining a fixed fidelity of $1 - \epsilon$.   
\par 

In a recent work~\cite{berta2021semidefinite}, the authors introduced an asymptotically converging semidefinite programming hierarchy $\{\SDP_{n}(\cN, M)\}_{n \in \NN}$ of outer bounds for $\mathrm{F}(\cN, M)$ (defined in Section~\ref{subsec:fidelity}). The hierarchy provides a measure of convergence speed and approaches the channel fidelity $\mathrm{F}(\cN, M)$ as the parameter $n$ tends to infinity. In particular, the authors showed that the inequality $0 \leq \mathrm{OPT}(\SDP_{n}(\cN, M)) - \mathrm{F}(\cN, M) \leq \frac{\poly(d)}{\sqrt{n}}$ holds, where $\mathrm{OPT}(\SDP_{n}(\cN, M))$ is the optimal value of the program $\SDP_{n}(\cN, M)$, and $d$ depends only on the input dimension, output dimension, $M$. This result implies that for any $\epsilon > 0$, estimating $\mathrm{F}(\cN, M)$ with an additive error of $\epsilon$ can be achieved by solving the semidefinite program $\SDP_{n}(\cN, M)$, where $n = \frac{\poly(d)}{\epsilon^2}$. However, it is important to note that the size of the matrix variables in $\SDP_{n}(\cN, M)$ grows exponentially with respect to $n$. Therefore, direct computation of the optimal value of the program $\SDP_{n}(\cN, M)$ is inefficient.

\paragraph*{Main result and techniques}
In this paper, we aim to propose an efficient way to compute the optimal value of the program $\SDP_{n}(\cN,M)$, and therefore, have an efficient approximate algorithm for computing the channel fidelity $\mathrm{F}(\cN,M)$, where $\cN$ is a quantum channel, and $M,n \in \mathbb{N}$. In particular, our main result is stated as follows.
\begin{thm}
For $M, n \in \NN_{\geq 1}$, let $\cN_{\bar{A}\to B}$ be a quantum channel, and let $d_{\bar{A}}$ denote the input dimension of $\cN_{\bar{A}\to B}$. The optimal value of the program $\SDP_{n}(\cN,M)$ can be computed in $\poly(n,d_{\bar{A}})$ time, for a fixed output dimension of the quantum channel $\cN_{\bar{A}\to B}$. As a direct consequence, the channel fidelity $\mathrm{F}(\cN,M)$ can be estimated with an additive error of $\epsilon$ in $\poly(1/\epsilon,d_{\bar{A}})$ time. 
\end{thm}

To achieve this goal, we further extend the technique developed in~\cite{litjens2017semidefinite, Polak_thesis} to exploit the symmetries of the semidefinite program $\SDP_{n}(\cN,M)$, which allows us to obtain an equivalent representation and solve it efficiently. Specifically, in Lemma~\ref{lem:invariant_matrices}, we demonstrate that the search space of the program can be restricted to an invariant subspace  $\End^{\mS_n}(\cH^n)$, which is obtained by taking a natural action of the symmetric group $\mS_n$ on the $\cH^n$ - the original search space. Utilizing tools from representation theory, particularly the representation theory of the symmetric group, we can construct a bijective linear map $\psi$ from $\End^{\mS_n}(\cH^n) $ to $\bigoplus_{i=1}^{t}\CC^{m_i \times m_i}$ for some integers $t$ and $m_i, i = 1,\dots,t$. This map ensures that for any $X \in \End^{\mS_n}(\cH^n)$, the matrix $\psi(X)$ takes on a block diagonal form with $\poly(n)$ number of blocks, i.e., $t = \poly(n)$, and each of which has a size bounded by a polynomial in $d_{\bar{A}}$ and $n$, i.e., $m_i = \poly(d_{\bar{A}},n)$ for all $i \in \{1,\dots,t \}$. 
\par
In addition, the bijective map $\psi$ preserves positive semidefiniteness, i.e., for any $X \in \End^{\mS_n}(\cH^n)$, $X \succeq 0$ if and only if $\psi(X) \succeq 0$. Therefore, for any $X \in \End^{\mS_n}(\cH^n)$, checking whether $X$ is a positive semidefinite matrix can be reduced to checking if the smaller matrices $[\psi(X)]_1,\dots,[\psi(X)]_t$ are all positive semidefinite, where $[\psi(X)]_i$ is the $i$-th block of $\psi(X)$ for $i = 1,\dots,t$.  
\begin{figure}[h!]
\begin{center}
    \includegraphics[width=8cm]{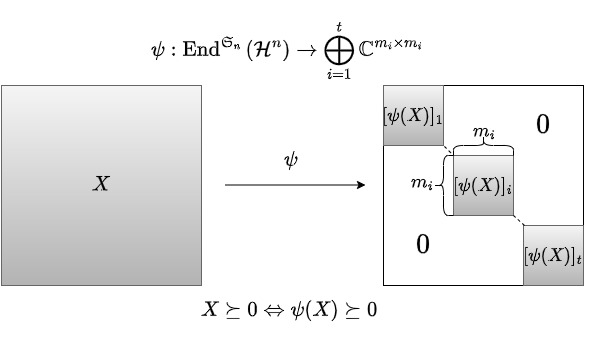}
\end{center}
\end{figure}

From the construction of $\psi$, we can build a transformation $\Phi$ that converts the program $\SDP_{n}(\cN, M)$ into an equivalent semidefinite program, denoted as $\Phi(\SDP_{n}(\cN, M))$, which has $\poly(d_{\bar{A}}, n)$ number of constraints, and each matrix variable has a polynomial size in $d_{\bar{A}}$ and $n$ (see Theorem~\ref{thm:convert_program}). Additionally, by utilizing the representation of the symmetric group and exploiting specific properties of the constraints involved in the program $\SDP_{n}(\cN, M)$, we show in Theorem~\ref{thm:transformation} that the transformation can be implemented in $\poly(d_{\bar{A}},n)$ time, while a direct implementation of the transformation would require exponential computation in $n$. Therefore, we can determine the optimal value of the program $\SDP_{n}(\cN, M)$ by solving the program $\Phi(\SDP_{n}(\cN, M))$, and this can be done in $\poly(d_{\bar{A}},n)$ time. The Table~\ref{table:reduce_dimension} shows the reduction in the size of the matrix variables in the program $\Phi(\SDP_{n}(\cN,2))$ compared to the program $\SDP_{n}(\cN,2)$ (for certain values of $n$) for arbitrary qubit quantum channels $\cN$.

\begin{table}[ht]
  	\begin{center}
  			\begin{tabular}{ |c|c|c|c| } 
  			\hline
  			$n$ & Matrix Size in $\SDP_{n}(\cN,2)$ & Maximum Matrix Size  in $\Phi(\SDP_{n}(\cN,2))$ & Number of Blocks \\
  			\hline 
  			$2$ &$64\times 64$& $40 \times 40$ & 2\\
  			\hline 
  			$3$ &$256\times 256$ &  $80 \times 80$& 3 \\ 
  			\hline
  			$4$ & $1024\times 1024$  & $180 \times 180$ & 5 \\
  			\hline
                $5$ & $4096\times 4096$  & $336 \times 336$ & 6 \\
  			\hline
                $6$ & $16384\times 16384$  & $560 \times 560$ & 9 \\
  			\hline
                $7$ & $65536\times 65536$  & $896 \times 896$ & 11 \\
  			\hline
                $8$ & $262144\times 262144$  & $1440 \times 1440$ & 15 \\
  			\hline
                $9$ & $1048576\times 1048576$  & $2160 \times 2160$ & 18 \\
  			\hline
                $10$ & $4194304\times 4194304$  & $3080 \times 3080$ & 23 \\
  			\hline
  		\end{tabular}
  		\caption{\label{table:reduce_dimension} Comparing the sizes of two semidefinite programs: $\SDP_{n}(\cN,2)$ vs. $\Phi(\SDP_{n}(\cN,2))$ for arbitrary qubit quantum channels $\cN$ }
  	\end{center}
\end{table}
In the case where the input dimension $d_{\bar{A}}$ of the channel $\mathcal{N}_{\bar{A} \to B}$ is fixed, we can consider another semidefinite programming hierarchy, which also provides an asymptotically converging sequence of outer bounds for $\mathrm{F}(\mathcal{N}, M)$ (see Section~\ref{subsec:fidelity} for more details). By applying the same approach developed in this work, we can show that the channel fidelity $\mathrm{F}(\mathcal{N}, M)$ can be estimated with an additive error of $\epsilon$ in $\text{poly}(1/\epsilon, d_B)$ time, where $d_B$ is the output dimension of the quantum channel $\mathcal{N}_{\bar{A} \to B}$.

\paragraph*{Related work}

Algorithmic aspects of optimal channel coding have been extensively studied in various settings. In classical channels, the problem was explored in~\cite{barman2017algorithmic} for point-to-point channels, and a generalization was investigated in~\cite{barman2020tight}. For entanglement-assisted classical channels, the research can be found in~\cite{berta2016quantum}, and for the multiple-access channel, the study is investigated in~\cite{fawzi2023multiple}, and for broadcast channels, it is discussed in~\cite{fawzi2023broadcast}. In terms of quantum channels, the problem has been examined in~\cite{leung2015power, berta2021semidefinite,kaur2021resource,Holdsworth_23} for general quantum channels and in~\cite{fawzi2019approximation} for classical-quantum channels. In addition, in recent work~\cite{Holdsworth_23}, authors exploited the unitary covariance symmetry of the identity channel to reduce the complexity of the $\SDP_n(\cN,M)$ when $n=2$ compared with the original form in~\cite{berta2021semidefinite}.
\par

  Exploiting symmetries to simplify convex programs, especially semidefinite programs, has been done for various problems and applications~\cite[Chapter 9]{Bachoc2012}. This approach has been seen in other surveys as well~\cite{vallentin2009symmetry,de2010exploiting}. Specifically, it has been applied to the mutually unbiased bases problem~\cite{gribling2021mutually} and the classical (quantum) capacity of quantum channels in quantum information theory~\cite{fawzi2022hierarchy}, quantum separability problem~\cite{doherty2004complete}. In coding theory, this method has been used to provide upper bounds for nonbinary codes~\cite{schrijver2005new,gijswijt2006new,laurent2007strengthened,litjens2017semidefinite,polak2019semidefinite}. Furthermore, this approach has also proven beneficial in various other fields and problems, such as queueing theory~\cite{brosch2021optimizing,polak2022symmetry}, crossing numbers for graphs~\cite{de2007reduction,brosch2022new}, truss topology optimization~\cite{bai2009exploiting}, polynomial optimization~\cite{gatermann2004symmetry,riener2013exploiting,raymond2018symmetric}, quantum channel discrimination~\cite{bergh2024parallelization}, compatibility of quantum marginals~\cite{huber2022refuting}, and quantum hypothesis testing~\cite{cheng2024sample}. Our work builds on the technical framework for exploiting symmetries under the symmetric group action, as developed in~\cite{litjens2017semidefinite,Polak_thesis}. This framework has been successfully extended to fundamental problems in quantum information theory~\cite{fawzi2022hierarchy,gribling2021mutually}. We further extend these methods to the problem of approximating quantum channel fidelity, yielding an efficient algorithm.

  %Our approach in this paper is to exploit symmetries under the symmetric group action, inspired by some prior works in~\cite{litjens2017semidefinite, gribling2021mutually, fawzi2022hierarchy}. 

\paragraph*{Organization } 
%\textcolor{blue}{write...}
The rest of this paper is organized as follows. Section~\ref{sec:preliminaries} provides an introduction to some fundamental notation and a brief overview of the problem of approximate quantum error correction, and presents tools for exploiting symmetries in semidefinite programs. Finally, our proof for the main result is presented in Section~\ref{sec:exloiting_symmetries}.
%%%%%%%%%%%%%%%%%%%%%%%%%%%%%%%%%%%%%%%%%%%%%%%%%%%      
\section{Preliminaries}
\label{sec:preliminaries}
\subsection*{Basic notation}  
Let $\cH$ be a finite dimensional complex Hilbert space; we denote by $\Lin(\cH)$ the set of linear operators on $\cH$, $\Pos(\cH)$ denotes the set of positive semidefinite operators on $\cH$, and $\cS(\cH):= \{\rho \in \Pos(\cH): \tr(\rho) = 1 \}$ is the set of density operators on $\cH$. For any two Hermitian operators $\rho, \sigma \in \Lin(\cH)$, we write $\sigma \succeq \rho $ if $\sigma - \rho \in \Pos(\cH)$. Let $A,B$ be finite dimensional complex Hilbert spaces, we denote $d_A,d_B$ as the dimension of $A$ and $B$, respectively. For $X \in \Pos(A\otimes B)$, we often explicitly indicate the quantum systems as a subscript by writing $X_{AB}$. The marginal on the subsystem $A$ is denoted $X_A =\tr_{B}(X_{AB}) \coloneqq \sum_{i}(\id_{A} \otimes \bra{i}_B)X_{AB}(\id_{A} \otimes \ket{i}_B)$, where $\{\ket{i}_B\}_i$ is an orthogonal basis of $B$ and $\id_A$ denotes the identity map on $\Lin(A)$. Let $\cbr{\ket{i}}_i$ and $\cbr{\ket{j}}_j$ be the standard bases for $A$ and $B$, respectively. We will use a correspondence between the linear operator in $\Lin(B,A)$ and vectors in $A\otimes B$, given by the linear map $\mathrm{vec}:\Lin(B,A)\rightarrow A\otimes B$, defined as $\mathrm{vec}\br{\ketbra{i}{j}}=\ket{i}\ket{j}$. In the case of $n$ copies of the same Hilbert space $\cH$, we denote by $\cH^{\otimes n} = \cH \otimes \cH \otimes \dots \otimes \cH$, we also use the notation $\cH_{1}^{n}$ or $\cH^n$ to indicate $\cH_1 \otimes \dots \otimes \cH_n$, where $\cH_i \cong \cH$ for all $i \in [n]$. For $i<j$, we denote by $\cH_{i}^{j} \coloneqq \cH_i \otimes \dots \otimes \cH_j$.  
\par
We denote by $\CP(A:B)$ the set of completely positive (CP) maps from $\Lin(A)$ to $\Lin(B)$. A \emph{quantum channel} $\cN_{A \to B}$ is a CP and trace-preserving linear map from $\Lin(A)$ to $\Lin(B)$. Let $A'$ be isomorphic to $A$ and $\ket{\Phi}_{AA'} = \frac{1}{\sqrt{d_A}}\sum_{i}\ket{i}_{A} \ket{i}_{A'}$ be the maximally entangled state. For a linear map $\cN_{A' \to B}$, we denote by $J_{AB}^{\cN} \in \Pos(A \otimes B)$ the corresponding \emph{Choi matrix} defined as $J_{AB}^{\cN} = (\id_A \otimes \cN)({\ketbra{\Phi}{\Phi}}_{AA'})$.
\par     
 For $n \in \mathbb{N}$, we use the notation $\mathfrak{S}_n$ to denote the symmetric group on $n$ symbols. This group consists of the permutations that can be performed on the $n$ symbols, and its group operation involves the composition of permutations. For every $\pi \in \mS_n$, we define the action of $\pi$ on $n$ copies of a finite dimensional Hilbert space $\cH^{\otimes n}$ as
\begin{align*}
	\pi \cdot (h_1 \otimes \dots \otimes h_n) = h_{\pi^{-1}(1)}\otimes \dots \otimes h_{\pi^{-1}(n)} \, , h_i \in \cH \, , \forall \pi \in \mathfrak{S}_n \, .
\end{align*}

For every $\pi \in \mS_n$, we denote by $U_{\cH^n}(\pi)$  a permutation matrix which corresponds to the action of $\pi$ on $\cH^{\otimes n}$. For any $\rho_{\cH^n} \in \Lin \left(\cH^{\otimes n} \right)$, we write
\begin{align*}
	U_{\cH^n}(\pi) \left( \rho_{\cH^n} \right) \coloneqq U_{\cH^n}(\pi) \rho_{\cH^n} U_{\cH^n}(\pi)^{*} \, ,
\end{align*}
where $U_{\cH^n}(\pi)^{*}$ is the conjugate transpose of $U_{\cH^n}(\pi)$. Moreover, the permutation matrices are real, thus $U_{\cH^n}(\pi)^{*} = U_{\cH^n}(\pi)^{T}$ for all $\pi \in \mS_n$. A multipartite operator $\rho_{\cH^n} \in \Lin \left(\cH^{\otimes n} \right)$ is called \emph{symmetric} if $\rho_{\cH^n} = U_{\cH^n}(\pi) \left(\rho_{\cH^n} \right)$ for all  $\pi \in \mS_n \, $.
% \par  
% Given two Hilbert spaces $\cA$ and $\cB$, a multipartite operator $\rho_{\cA \cH^n \cB} \in \Lin \left( \cA \otimes \cH^{\otimes n} \otimes \cB \right)$ is said to be \emph{symmetric with respect to} $\cA$ and $\cB$, if it is invariant under permutation of $\cH$-systems while keeping $\cA$ and $\cB$ fixed. In particular,
% \begin{align*}
% 	\rho_{\cA \cH^n \cB} = \left(\id_{\cA} \otimes U_{\cH^n}(\pi)\otimes \id_{\cB} \right) \left( \rho_{\cA \cH^n \cB}  \right) \enspace \text{ for all } \pi \in \mS_n \, .
% \end{align*}
%% 

%%%%%%%%%%%%%%%
\subsection{Approximate quantum error correction}
\label{subsec:fidelity}
In this section, we briefly present the mathematical setting of approximate quantum error correction. For further information, we refer to~\cite{berta2021semidefinite,borderi2022finetti}.
 
Let $A,B,\bar{A},\bar{B}$ be finite dimensional complex Hilbert spaces. Let $\cN_{\bar{A}\to B}$ ($\cN$ in short) be a quantum channel. The quantum channel fidelity (or short channel fidelity) for message dimension $M \in \NN$ is defined as
\begin{align*}
\mathrm{F}(\cN,M) \coloneqq \max_{\cD_{B \to \bar{B}} \, , \cE_{A \to \bar{A}}} & \,\, F \left( \Phi_{\bar{B}R}, ((\cD_{B \to \bar{B}} \circ \cN_{\bar{A}\to B} \circ \cE_{A \to \bar{A}})\otimes \cI_{R})(\Phi_{AR}) \right) \\
& \text{s.t. } \space \cD_{B \to \bar{B}} \, , \cE_{A \to \bar{A}} \, \, \, \, \text{quantum channels} \, ,  
\end{align*}
where $F(\rho,\sigma) \coloneqq \| \sqrt{\rho}\sqrt{\sigma} \|_{1}^{2}$ denotes the fidelity, $\Phi_{AR}$ denotes the maximally entangled state on $AR$, and $M = d_{A} = d_{\bar{B}} = d_{R}$. 

Using the Choi-Jamiolkowski isomorphism, the authors in~\cite{berta2021semidefinite} demonstrated that the channel fidelity can be reformulated as a bilinear optimization problem as follows. 
\begin{equation}
     \label{eq:channel_fidelity}
    \begin{split}
        \mathrm{F}(\cN,M) = \max_{E_{A\bar{A}},D_{B\bar{B}}} \, \, &d_{\bar{A}}d_{B} \cdot \tr \left[  (J_{\bar{A}B}^{\cN} \otimes \Phi_{A\bar{B}})(E_{A\bar{A}} \otimes D_{B\bar{B}} ) \right] \\
& \text{s.t. } \space E_{A\bar{A}} \succeq 0, \, , D_{B\bar{B}} \succeq 0 \, , \\
&\space E_{A} = \frac{\id_A}{d_{A}}, \, D_{B} = \frac{\id_B}{d_B} \, .
    \end{split}
\end{equation}
\par
By the linearity of the objective function, we can write the above program in an equivalent form that optimizes over a convex hull of feasible solutions~\cite{berta2021semidefinite} as follows.
\begin{equation} \label{eq:convex_form}
    \begin{split}
        \mathrm{F}(\cN,M) = \max \, \, &d_{\bar{A}}d_{B} \cdot \tr \left[  \left( J_{\bar{A}B}^{\cN} \otimes \Phi_{A\bar{B}}\right) \left( \sum_{i \in I} p_i E_{A\bar{A}}^{i} \otimes D_{B\bar{B}}^{i} \right) \right] \\
& \text{s.t. } p_i \geq 0 \, , \, \forall i \in I, \, \, \sum_{i \in I}p_i = 1 \, , \\
&\space E_{A\bar{A}}^{i} \succeq 0, \, D_{B\bar{B}}^{i} \succeq 0  \; \; \forall i \in I \, ,\\
&\space E_{A}^{i} = \frac{\id_A}{d_{A}}, \, D_{B}^{i} = \frac{\id_B}{d_B} \; \; \forall i \in I \, .
    \end{split}
\end{equation}
Program~\eqref{eq:convex_form} can be viewed as the maximization of a linear function over a subset of quantum bipartite states known as \emph{separable quantum states}. Approximating the set of separable states within the set of bipartite states is a hard problem in quantum information, as showed in \cite{Gharibian_10}. However, it is possible to approximate the set of separable states using a semidefinite programming hierarchy~\cite{doherty2002distinguishing,doherty2004complete}, and the convergence of this approximation can be analyzed with the \emph{finite quantum de Finetti theorem}~\cite{christandl2007one}. Similarly, in~\cite{berta2021semidefinite}, the authors considered the problem of approximating feasible solutions of Program~\eqref{eq:convex_form} and introduced a semidefinite programming hierarchy $\SDP_{n}(\cN,M)$ for this purpose as follows.   

\begin{equation} \label{eq:sdp_hierarchy} \tag{$\mathrm{SDP}_{n}(\cN,M)$}
	\begin{split}
		\mathrm{OPT}(\SDP_{n}(\cN,M)) \coloneqq &\max_{\rho_{A\bar{A}(B\bar{B})_{1}^{n}}} \; \; d_{\bar{A}}d_{B} \cdot \tr \left[(J_{\bar{A}B_1}^{\cN}\otimes \Phi_{A\bar{B}_1} )\rho_{A\bar{A}B_1 \bar{B}_1} \right]  \\
		&\text{s.t.} \; \;  \rho_{A\bar{A}(B\bar{B})_{1}^{n}} \succeq 0 \, \, , \tr \left[ \rho_{A\bar{A}(B\bar{B})_{1}^{n}}\right] = 1 \, ,  \\
		&  \rho_{A\bar{A}(B\bar{B})_{1}^{n}} =  \left( \id_{A\bar{A}} \otimes U_{(B\bar{B})_{1}^{n}}(\pi) \right) \left( \rho_{A\bar{A}(B\bar{B})_{1}^{n}}\right) \, \, \forall \pi \in \mS_n \, , \\
		 &  \rho_{A(B\bar{B})_{1}^{n}} = \frac{\id_A}{d_A}\otimes \rho_{(B\bar{B})_{1}^{n}} \, , \\
		& \rho_{A\bar{A}(B\bar{B})_{1}^{n-1}B_n} = \rho_{A\bar{A}(B\bar{B})_{1}^{n-1}} \otimes \frac{\id_{B_n}}{d_B} \, .
	\end{split}
\end{equation}  
Note that we identified $B_1 \coloneqq B$ in Program~\ref{eq:sdp_hierarchy} and recall that $B_1 \cong B_2 \cong \dots \cong B_n$  and $\bar{B}_1 \cong \bar{B}_2 \cong  \dots \cong \bar{B}_n$. In addition, by introducing novel finite quantum de Finetti theorems that allow imposing linear constraints on the approximating states,~\cite{berta2021semidefinite} provided a quantification of the convergence of Program~\ref{eq:sdp_hierarchy} to channel fidelity $\mathrm{F}(\cN, M)$. 

%\begin{theorem}[\cite{berta2021semidefinite}] \label{thm:hierarchy_bound}
%Let $\cN_{\bar{A}\to B}$ be a quantum channel and $n,M \in \NN$. Then, we have
%\begin{align*}
%0 \leq \mathrm{OPT}(\SDP_{n}(\cN,M)) - \mathrm{F}(\cN,M) \leq \frac{\poly(d)}{\sqrt{n}} \, ,
%\end{align*}
%where $d = \max \{d_A= d_{\bar{B}} = M,d_{\bar{A}},d_B \}$. 
%\end{theorem}
%%%%
\begin{proposition}[\cite{berta2021semidefinite}]
	\label{pro:hierarchy_bound}
	Let $\cN_{\bar{A}\to B}$ be a quantum channel and $n,M \in \NN$. Then, we have
	\begin{align*}
		0 \leq \mathrm{OPT}(\SDP_{n}(\cN,M)) - \mathrm{F}(\cN,M) \leq \frac{\poly(d)}{\sqrt{n}} \, ,
	\end{align*}
	where $d = \max \{d_A= d_{\bar{B}} = M,d_{\bar{A}},d_B \}$. 
\end{proposition}

%%%%
Noting that instead of extending the $B$-systems we could alternatively extend the $A$-systems, which leads to the following asymptotically converging semidefinite program hierarchy $\overline{\SDP}_{n}(\cN,M)$ for approximating $\mathrm{F}(\cN,M)$~\cite{berta2021semidefinite}. 

\begin{equation} \label{eq:sdp_hierarchy_Asystem} \tag{$\overline{\mathrm{SDP}}_{n}(\cN,M)$}
	\begin{split}
		\mathrm{OPT}(\overline{\SDP}_{n}(\cN,M)) \coloneqq &\max_{\rho_{(A\bar{A})_{1}^{n}B\bar{B}}} \; \; d_{\bar{A}}d_{B} \cdot \tr \left[(J_{\bar{A}B_1}^{\cN}\otimes \Phi_{A\bar{B}_1} )\rho_{A_1\bar{A}_1B \bar{B}} \right]  \\
		&\text{s.t.} \; \;  \rho_{(A\bar{A})_{1}^{n}B\bar{B}} \succeq 0 \, \, , \tr \left[ \rho_{(A\bar{A})_{1}^{n}B\bar{B}}\right] = 1 \, ,  \\
		&  \rho_{(A\bar{A})_{1}^{n}B\bar{B}} =  \left(U_{(A\bar{A})_{1}^{n}}(\pi) \otimes \id_{B\bar{B}} \right) \left( \rho_{(A\bar{A})_{1}^{n}B\bar{B}} \right) \, \, \forall \pi \in \mS_n \, , \\
		 &  \rho_{(A\bar{A})_{1}^{n}B} = \rho_{(A\bar{A})_{1}^{n}} \otimes \frac{\id_{B}}{d_{B}} \, , \\
		& \rho_{(A\bar{A})_{1}^{n-1}A_nB\bar{B}} = \frac{\id_{A_n}}{d_A} \otimes \rho_{(A\bar{A})_{1}^{n-1}B\bar{B}} \, .
	\end{split}
\end{equation}  

The Program~\ref{eq:sdp_hierarchy} and Program~\ref{eq:sdp_hierarchy_Asystem} are non-equivalent. However, the convergence guarantees in Proposition~\ref{pro:hierarchy_bound} still hold, i.e., $0 \leq \mathrm{OPT}(\overline{\SDP}_{n}(\cN,M)) - \mathrm{F}(\cN,M) \leq \frac{\poly(d)}{\sqrt{n}} \, ,$ where $d = \max \{d_A= d_{\bar{B}} = M,d_{\bar{A}},d_B\}$.
\par 
In this work, given $M = d_A = d_{\bar{B}} \in \NN$, we consider the problem of computing $\SDP_n(\cN,M)$ where the output dimension $d_{B}$ of $\cN_{\bar{A}\to B}$ is fixed. In this case, following  Proposition~\ref{pro:hierarchy_bound}, one has 
\begin{align*}
    0 \leq \mathrm{OPT}(\SDP_{n}(\cN,M)) - \mathrm{F}(\cN,M) \leq \frac{\poly(d_{\bar{A}})}{\sqrt{n}} \, ,
\end{align*}
where $d_{\bar{A}}$ is the input dimension of the channel $\cN$. In particular, $\poly(d_{\bar{A}})$ is at most $\bigO(d_{\bar{A}}\sqrt{\log d_{\bar{A}}})$ following~\cite[Theorem 4.2.1]{borderi2022finetti}. Therefore, by choosing $n = \frac{\poly(d_{\bar{A}})}{\epsilon^2}$, we can estimate $\mathrm{F}(\cN, M)$ with an accuracy of $\epsilon$ by solving $\SDP_{n}(\cN, M)$. However, the size of the matrix variables in $\SDP_{n}(\cN, M)$ grows exponentially with $n$. Thus, solving $\SDP_{n}(\cN, M)$ directly would require exponential time in terms of $n$. As a consequence, estimating $\mathrm{F}(\cN, M)$ with an accuracy of $\epsilon$  requires $\exp \left ( 1/\epsilon, d_{\bar{A}}\right)$ time by solving the Program~\ref{eq:sdp_hierarchy} directly.  In this paper, we provide an efficient method to compute the optimal value of the program $\SDP_{n}(\cN, M)$, which can be done in $\poly(d_{\bar{A}},n)$ time.

\begin{remark}
    For a fixed $M$ and the output dimension $d_{B}$ of $\cN_{\bar{A}\to B}$, the channel fidelity $\mathrm{F}(\cN,M)$ can be determined in exponential time in terms of $d_{\bar{A}}$ by solving Program~\eqref{eq:channel_fidelity} using a brute-force algorithm, as it forms a constrained polynomial optimization problem. Some methods for solving constrained polynomial optimization problems can be found at~\cite{lasserre2015introduction}. 
\end{remark}
\begin{remark} 
Noting that our method in this work can be adapted to the case where $M$ and the input dimension $d_{\bar{A}}$ of $\cN_{\bar{A}\to B}$ are fixed. In this scenario, the same strategy can be employed to exploit the symmetries in Program~\ref{eq:sdp_hierarchy_Asystem}, yielding $\text{poly}(n,d_{B})$ time complexity for solving $\overline{\text{SDP}}_{n}(\cN,M)$.
\end{remark}

%%%%%%%%%%%%%%%%%%%%%%%%%%%%

\subsection{Elementary Representation Theory} \label{sec:Tools}

In this subsection, we present the essential mathematical foundation for exploiting symmetries in a semidefinite program to represent the program effectively. Most of the concepts and notation in this section can be found in~\cite{litjens2017semidefinite,fawzi2022hierarchy}, but we include them here for self-containedness. For further information, readers interested in this topic can refer to references such as~\cite[Chapter 9]{Bachoc2012}.

%----------------------------------------------------

\subsection*{Representation theory}

   We recall some basic facts and notions in the representation theory of finite groups. For further information, we refer the reader to Refs.~\cite{rep_1977} and~\cite{Fulton1991}. Let $G$ be a finite group and $\cH$ be a finite dimensional complex Hilbert space. A group homomorphism $\varrho: G \to \GL(\cH)$ is called a \emph{linear representation} of $G$ on $\cH$, where $\GL(\cH)$ is the general linear group on $\cH$. The space $\cH$ is called a \emph{$G$-module}. For $v \in \cH$ and $g\in G$, we write $g\cdot v$ as shorthand for $\varrho(g)v$. For $X\in\Lin(\cH)$, the action of $g\in G$ on $X$ is given by $\varrho(g)X\varrho(g)^{*}$.
     
   A representation $\varrho: G\to \GL(\cH)$ of $G$ is called \emph{irreducible} if for any subspace $\cH'$ of $\cH$ then $ \{g\cH': g\in G \} \subsetneq \cH'$, in other words, it contains no proper submodule $\cH'$ of $\cH$ such that $g \cH' \subseteq \cH'$. Let $\cH$ and $\cH'$ be $G$-modules, a \emph{$G$-equivariant map} from $\cH$ to $\cH'$ is a linear map $\phi: \cH \to \cH'$  such that $g \cdot \phi(v) = \phi(g \cdot v)$ for all $g \in G, v\in \cH$.  Two $G$-modules $\cH$ and $\cH'$ are called \emph{$G$-isomorphic}, write $\cH \cong \cH'$, if there is a bijective equivariant map from $\cH$ to $\cH'$. We denote by $\End^{G}(\cH)$, the set of all $G$-equivariant maps from $\cH$ to $\cH$, i.e.,  
\begin{align*} 
	\End^{G}(\cH) = \{T \in \Lin(\cH): T(g\cdot v) = g\cdot T(v), \forall v \in \cH, g \in G \}.  
\end{align*}
Note that $\End^{G}(\cH)$ is known to form a matrix $*$-algebra~\cite{Bachoc2012}, that is, a set of complex matrices that is closed under addition, scalar multiplication, matrix multiplication, and taking the conjugate transpose. 
\par 

Let $G$ be a finite group acting on a finite dimensional complex vector space $\cH$. Then the space $\cH$ can be decomposed into a direct sum of subspaces as $\cH = \cH_1 \oplus \dots \oplus \cH_t$, where $\cH_1,\dots,\cH_t$ are $G$-modules and called the \emph{$G$-isotypical component}. In more detail, every $\cH_i$ is a direct sum of irreducible $G$-modules denoted as $\cH_{i,1} \oplus \dots \oplus \cH_{i,m_i}$ and note that two irreducible $G$-modules $\cH_{i,j}$ and $\cH_{i',j'}$ are isomorphic (i.e., $\cH_{i,j} \cong \cH_{i',j'}$) if and only if $i=i'$. The tuple $(m_1,\dots,m_t)$ are called \emph{multiplicities} of the corresponding irreducible representations.
\par  
For each $i \in [t]$ and $j \in [m_i]$, let $u_{i,j} \in \cH_{i,j}$ be a nonzero vector such that for each $i$ and all $j,j' \in [m_i]$, there is a bijective $G$-equivariant map from $\cH_{i,j}$ to $\cH_{i,j'}$ that maps $u_{i,j}$ to $u_{i,j'}$. For $i \in [t]$, we define a matrix $U_i$ as $[u_{i,1},\dots,u_{i,m_i}]$, with $u_{i,j}$ forming the $j$-th column of $U_i$. 
The matrix set $\{U_1,\dots,U_t\}$ obtained in this way is called a \emph{representative matrix set} for the action of $G$ on $\cH$. The columns of the matrices $U_i$ can be viewed as elements of the dual space $\cH^{*}$ (by taking the standard inner product). Then each $U_i$ is an ordered set of linear functions on $\cH$. Since $\cH_{i,j}$ is the linear space spanned by $G \cdot u_{i,j}$ (for each $i,j$), we have 
\begin{align*}
    \cH = \bigoplus_{i=1}^{t}\bigoplus_{j=1}^{m_i} \CC G \cdot u_{i,j} \; ,
\end{align*}
 where $\CC G=\cbr{\sum_{g\in G}\alpha_g g: \alpha_g\in \CC}$ denotes the complex group algebra of $G$.  Furthermore, we have
 \begin{align} \label{eq:dim_END vs multiplicites}
 	\dim \End^{G}(\cH) = \dim \End^{G}\left(  \bigoplus_{i=1}^{t}\bigoplus_{j=1}^{m_i} \cH_{i,j}  \right) = \sum_{i=1}^{t} m_{i}^2 \, .
 \end{align}

Note that with the action of the finite group $G$ on the space $\cH$, any inner product $\< \, , \>$ on $\cH$ gives rise to a $G$-invariant inner product $\< \, ,\>_{G}$ on $\cH$ via the rule $\<x,y\>_{G} \coloneqq \frac{1}{|G|}\sum_{g \in G}\<g \cdot x,g\cdot y\>$. Let $\<\, ,\>$ be a $G$-invariant inner product on $\cH$ and $\{U_1,\dots,U_t\}$ be a representative matrix set for the action of $G$ on $\cH$. Consider the linear map $\psi : \End^{G}(\cH) \to \bigoplus_{i=1}^{t}\CC^{m_i \times m_i}$ defined as
\begin{align}
	\label{eq:block_diagonal_matrix}
	\psi(X) \coloneqq \bigoplus_{i=1}^{t} \left( \<Xu_{i,j'},u_{i,j}\>\right)_{j,j'=1}^{m_i} \;,\; \forall X\in \End^{G}(\cH) \enspace.
\end{align} 
For $i \in [t]$ and $X\in \End^{G}(\cH)$, we denote the matrix $\left( \<Xu_{i,j'},u_{i,j}\>\right)_{j,j'=1}^{m_i}$ corresponding to the $i$-th block of $\psi(X)$ by $\llbracket \psi(X) \rrbracket_i$. 
\begin{lemma}[Proposition 2.4.4, \cite{Polak_thesis}]
	\label{lemma:psd_preserving}
	The linear map $\psi$ of Eq.~\eqref{eq:block_diagonal_matrix} is bijective and for every $X \in \End^{G}(\cH)$, we have $X \succeq 0$ if and only if $\psi(X) \succeq 0$. Moreover, there is a unitary matrix $U$ such that
	\begin{align*}
			U^{*}XU = \bigoplus_{i=1}^{t} \bigoplus_{j=1}^{m_i} \llbracket \psi(X) \rrbracket_i \;,\; \forall X\in \End^{G}(\cH) \enspace ,
	\end{align*}
	where $m_{i} = \dim(\cH_{i,1})$, for every $i \in [t]$.
\end{lemma}
Lemma~\ref{lemma:psd_preserving} plays a crucial role in our symmetry reductions. Noting that $\dim (\End^{G}(\cH)) = \sum_{i=1}^{t}m_{i}^{2}$ can be significantly smaller than the dimension of $\cH$. Furthermore, thanks to this lemma, we can simplify the task of verifying whether a matrix $X\in \End^{G}(\cH)$ is positive semidefinite. This simplification involves checking if the smaller $m_i\times m_i$ matrices $\llbracket \psi(X) \rrbracket_i$ are positive semidefinite for each $i\in[t]$. The following result, which was mentioned in~\cite{Polak_thesis}, helps construct a representative set of direct product groups.

\begin{lemma} \label{lem:direct_groups_representative}
	Let $G_1$ and $G_2$ be two finite groups. Let $\{U_{1}^{(1)},\dots,U_{k_1}^{(1)}\}$ and $\{U_{1}^{(2)},\dots,U_{k_2}^{(2)} \}$ be the representative matrix sets that correspond to the action of $G_1$ and $G_2$ on $\cH_1$ and $\cH_2$, respectively. Then 
	\begin{align} 
		\{ U_{i}^{(1)} \otimes U_{j}^{(2)} : i=1,\dots,k_1, \, j = 1,\dots,k_2 \} \label{eq:product_representative}
	\end{align}
	is representative matrix set for the action of $G_1 \times G_2$ on $\cH_1 \otimes \cH_2$. 
\end{lemma}
\begin{proof}
	The proof can be found in Appendix~\ref{Apen:representative_set}. 
\end{proof}

\subsection*{Representation theory of the symmetric group}

Fix $n \in \NN$ and a finite-dimensional vector space $\cH$ with $\dim(\cH) = d$. We consider the natural action of the symmetric group $\mathfrak{S}_n$ on $\cH^{\otimes n}$ by permuting the indices, i.e.,
\begin{equation}
     \label{eq:natural_action}
    \pi \cdot (h_1 \otimes \dots \otimes h_n) = h_{\pi^{-1}(1)}\otimes \dots \otimes h_{\pi^{-1}(n)} \, , h_i \in \cH \, , \forall \pi \in \mathfrak{S}_n \, .
\end{equation}
	
Based on representation theory of the symmetric group, we describe a representative set for the action of $\mathfrak{S}_n$ on $\cH^{\otimes n}$. Many of the concepts and notation are based on the method for symmetry
reduction from~\cite{litjens2017semidefinite}. The concepts we describe in this section will be used throughout this paper

\par 

A \emph{partition} $\lambda$ of $n$ is a sequence $(\lambda_1,\dots,\lambda_d)$ of natural numbers with $\lambda_1 \geq \dots \geq \lambda_d>0$ and $\lambda_1+\dots+\lambda_d = n$. The number $d$ is called the \emph{height} of $\lambda$. We write $\lambda \vdash_{d}n$ if $\lambda$ is a partition of $n$ with at most $d$ parts. Let $\mathrm{Par}(d,n) \coloneqq \{ \lambda\,:\,\lambda \vdash_{d} n \}$. The \emph{Young shape} $Y(\lambda)$ of $\lambda$ is the set
\begin{align*}
	Y(\lambda)\coloneqq \{(i,j) \in \NN^2: 1 \leq j \leq d, 1 \leq i \leq \lambda_j \} \, .
\end{align*}

Following the French notation~\cite{procesi2007lie}, for an index $j_0 \in [d]$, the $j_0$-th \emph{row} of $Y(\lambda)$ is set of elements $(i,j_0)$ in $Y(\lambda)$. Similarly, fixing an element $i_0 \in [\lambda_1]$, the $i_0$-th \emph{column} of $Y(\lambda)$ is set of elements $(i_0,j)$ in $Y(\lambda)$. We label the elements in $Y(\lambda)$ from $1$ to $n$ according the lexicographic order on their positions. Then the \emph{row stabilizer} $R_{\lambda}$ of $\lambda$ is the group of permutations $\pi$ of $Y(\lambda)$ with $\pi(L) = L$ for each row $L$ of $Y(\lambda)$. Similarly, the \emph{column stabilizer} $C_{\lambda}$ of $\lambda$ is the group of permutations $\pi$ of $Y(\lambda)$ with $\pi(L) = L$ for each column $L$ of $Y(\lambda)$. 
\par 
For $\lambda \vdash_{d} n$, a \emph{$\lambda$-tableau} is a function $\tau: Y(\lambda) \to \NN$. A $\lambda$-tableau is \emph{semistandard} if the entries are non-decreasing in each row and strictly increasing in each column. Let $T_{\lambda,d}$ be the collection of semistandard $\lambda$-tableaux with entries in $[d]$. We write $\tau \sim \tau'$ for $\lambda$-tableaux $\tau,\tau'$ if $\tau' = \tau r$ for some $r \in R_{\lambda}$. Let $e_1,\dots,e_d$ be the standard basis of $\cH$. For any $\tau \in T_{\lambda,d}$, define $u_{\tau}\in \cH^{\otimes n}$ as
\begin{align}
	u_{\tau} \coloneqq \sum_{\tau' \sim \tau}\sum_{c \in C_{\lambda}} \mathrm{sgn}(c) \bigotimes_{y \in Y(\lambda)}e_{\tau'(c(y))} \, .
\end{align}
Here the Young shape $Y(\lambda)$ is ordered by concatenating its rows. Then the matrix set
\begin{align}
	\label{eq:symmetric_representative_set}
	\{U_{\lambda}: \lambda \vdash_{d} n \} \, \, \text{ with } U_{\lambda} = [u_{\tau}: \tau \in T_{\lambda,d}]
\end{align}

is a representative matrix set for the natural action of $\mathfrak{S}_n$ on $\cH^{\otimes n}$ \cite[Section 2.1]{litjens2017semidefinite}. Moreover, we have
\begin{align}
	\label{eq:number_partitions}
	|\mathrm{Par}(d,n)| \leq (n+1)^d \text{ and } |T_{\lambda,d}| \leq (n+1)^{d(d-1)/2} \, \, , \forall \lambda \in \mathrm{Par}(d,n) \, .
\end{align}
From Eq.~\eqref{eq:number_partitions}, if $d$ is fixed, then both $|\mathrm{Par}(d,n)|$ and $|T_{\lambda,d}|$ are bounded by $\poly(n)$.  
%%%%%%%%%%%%%%%%%%%%%%%%%%%%%%%%%%%%%%%%%%%%%%%%%%%%%%

\section{Efficient approximation of channel fidelity}
\label{sec:exloiting_symmetries}

%As we have seen, the size of Program~\ref{eq:sdp_hierarchy} grows exponentially with $n$. In this section, we will show how to use the symmetries of Program~\ref{eq:sdp_hierarchy} to simplify this optimization problem and solve it in polynomial time in $n$.
%\par 
Given Hilbert space $\cH$, we follow the notation in Section~\ref{sec:preliminaries} and consider the action of symmetric group $\mS_n$ on $\cH^{\otimes n}$ defined in Eq.~\eqref{eq:natural_action}. Throughout this paper, we work with spaces of the form $\cA \otimes \cH^{\otimes n} \otimes \cB$ for some Hilbert spaces $\cA,\cB$. A multipartite operator $\rho_{\cA \cH^n \cB} \in \Lin \left( \cA \otimes \cH^{\otimes n} \otimes \cB \right)$ is said to be \emph{symmetric with respect to} $\cA$ and $\cB$, if it is invariant under permutation of $\cH$-systems while keeping $\cA$ and $\cB$ fixed. In particular, 
\begin{align*}
    \rho_{\cA \cH^n \cB} &= \left(\id_{\cA} \otimes U_{\cH^n}(\pi)\otimes \id_{\cB} \right) \left( \rho_{\cA \cH^n \cB}  \right) \\
                         &\coloneqq \left(\id_{\cA} \otimes U_{\cH^n}(\pi)\otimes \id_{\cB} \right) \rho_{\cA \cH^n \cB} \left(\id_{\cA} \otimes U_{\cH^n}(\pi)^{* }\otimes \id_{\cB} \right) \;  \forall  \pi \in \mS_n \, .
\end{align*}

\par 
The invariant subspace under the action of $\mS_n$ is given by
\begin{align*}
	\End^{\mS_n} \left(\cA \otimes \cH^{\otimes n} \otimes \cB \right) \coloneqq \left\{ \rho \in \Lin \left( \cA \otimes \cH^{\otimes n} \otimes \cB \right):  \rho = \left(\id_{\cA} \otimes U_{\cH^n}(\pi)\otimes \id_{\cB} \right) \left( \rho  \right) \enspace \text{ for all } \pi \in \mS_n \,  \right\} \, .
\end{align*}
We first show in the following lemma that the search space of Program~\ref{eq:sdp_hierarchy} can be restricted to the invariant subspaces that arising under the action of $\mS_n$.   

%%%%%%%%
\begin{lemma} \label{lem:invariant_matrices}
For $n \in \NN_{\geq 1}$, if $\rho_{A\bar{A}(B\bar{B})_{1}^n} \in \End^{\mS_n} \left(  A\bar{A} \otimes (B\bar{B})^{\otimes n} \right)$, then
\begin{align}
&\rho_{A(B\bar{B})_{1}^{n}} \, , \, \frac{\id_A}{d_A} \otimes \rho_{(B\bar{B})_{1}^{n}} \in \End^{\mS_n} \left(A \otimes (B\bar{B})^{\otimes n} \right) \label{eq:invariant1}  \, \, ,\\
& \rho_{A\bar{A}(B\bar{B})_{1}^{n-1}B_n} \, , \, \rho_{A\bar{A}(B\bar{B})_{1}^{n-1}} \otimes \frac{\id_{B_n}}{d_{B}} \in \End^{\mS_{n-1}} \left(A\bar{A} \otimes (B\bar{B})^{\otimes (n-1)}\otimes B_n  \right) \label{eq:invariant2} \, .
\end{align}
\end{lemma}
%\begin{proof}
%The proof can be found in Appendix~\ref{Apen:lem:invariant}.
%\end{proof}
\begin{proof}
	To prove Claim~\eqref{eq:invariant1}, we express $\rho_{A\bar{A}(B\bar{B})_{1}^{n}}$ in terms of the standard bases of $\Lin(A)$ and $\Lin(\bar{A})$. Specifically, we write
	\begin{align*}
		\rho_{A\bar{A}(B\bar{B})_{1}^{n}} = \sum_{\substack{i,j \in [d_{A}] \\ x,y \in [d_{\bar{A}}]} }\ketbra{i}{j} \otimes \ketbra{x}{y} \otimes \rho^{i,j,x,y} \, .
	\end{align*}
	Taking the partial trace over $\bar{A}$, we obtain
	\begin{align*}
		\rho_{\bar{A}(B\bar{B})_{1}^{n}} \coloneqq \tr_{\bar{A}}(\rho_{A\bar{A}(B\bar{B})_{1}^{n}}) = \sum_{\substack{i,j \in [d_A] \\ x \in [d_{\bar{A}}]}}\ketbra{i}{j} \otimes \rho^{i,j,x,x} \, .
	\end{align*}
	Moreover, since $\rho_{A\bar{A}(B\bar{B})_{1}^{n}} \in \End^{\mS_n} \left(  A\bar{A} \otimes (B\bar{B})^{\otimes n} \right)$, we have
	\begin{align*}
		\rho_{A\bar{A}(B\bar{B})_{1}^{n}} = \sum_{\substack{i,j \in [d_{A}] \\ x,y \in [d_{\bar{A}}]} }\ketbra{i}{j} \otimes \ketbra{x}{y} \otimes U_{(B\bar{B})_{1}^{n}}(\pi)\rho^{i,j,x,y}U_{(B\bar{B})^{n}}(\pi)^{*} \text{ for any } \pi \in \mS_n \, .
	\end{align*}
	Therefore,   
	\begin{align*}
		\rho_{\bar{A}(B\bar{B})_{1}^{n}} &\coloneqq \tr_{\bar{A}}(\rho_{A\bar{A}(B\bar{B})_{1}^{n}})  \\
		&= \sum_{\substack{i,j \in [d_A]\\ x \in [d_{\bar{A}}]}}\ketbra{i}{j} \otimes U_{(B\bar{B})_{1}^{n}}(\pi)\rho^{i,j,x,x}U_{(B\bar{B})_{1}^{n}}(\pi)^{*} \\
		& = 
		\left(\id_A \otimes U_{(B\bar{B})_{1}^{n}}(\pi) \right) \left(\sum_{\substack{i,j \in [d_A] \\ x \in [d_{\bar{A}}]}}\ketbra{i}{j} \otimes \rho^{i,j,x,x} \right) \left(\id_A \otimes U_{(B\bar{B})_{1}^{n}}(\pi)^{*} \right) \\
		&= \left(\id_A \otimes U_{(B\bar{B})_{1}^{n}}(\pi) \right) \left(\tr_{\bar{A}}(\rho_{A\bar{A}(B\bar{B})_{1}^{n}}) \right) \left(\id_A \otimes U_{(B\bar{B})_{1}^{n}}(\pi)^{*} \right) \,.
	\end{align*}
	Since this equality holds for any $\pi \in \mS_n$, we conclude that
	\begin{align*}
		\rho_{\bar{A}(B\bar{B})_{1}^{n}} \coloneqq \tr_{\bar{A}}(\rho_{A\bar{A}(B\bar{B})_{1}^{n}}) \in \End^{\mS_n} \left(A \otimes (B\bar{B})^{\otimes n} \right) \, .
	\end{align*}
	
	Similarly, we also have
	\begin{align*}
		\frac{\id_A}{d_A} \otimes \rho_{(B\bar{B})_{1}^{n}} \in \End^{\mS_n} \left(A \otimes (B\bar{B})^{\otimes n} \right) \, .
	\end{align*}
	This proves the Claim~\eqref{eq:invariant1}.
	\par 
	
	To prove Claim~\eqref{eq:invariant2}, we write
	\begin{align*}
		\rho_{A\bar{A}(B\bar{B})_{1}^{n}} = \sum_{\substack{i,j \in [d_{A\bar{A}}] \\ u,v \in [d_{B_n}] \\ x,y \in [d_{\bar{B}_n}]}} \ketbra{i}{j} \otimes \rho^{i,j,u,v,x,y} \otimes \ketbra{u}{v} \otimes \ketbra{x}{y} \, .
	\end{align*}
	Thus,
	\begin{align*}
		\rho_{A\bar{A}(B\bar{B})_{1}^{n-1}B_n} \coloneqq \tr_{\bar{B}_n} \left(\rho_{A\bar{A}(B\bar{B})_{1}^{n}} \right) = \sum_{\substack{i,j \in [d_{A\bar{A}}] \\ u,v \in [d_{B_n}] \\ x\in [d_{\bar{B}_n}]}} \ketbra{i}{j} \otimes \rho^{i,j,u,v,x,x} \otimes \ketbra{u}{v} \, .
	\end{align*}
	For any $\pi \in 
	\mS_{n-1}$, the permutation matrix  $U_{(B\bar{B})_{1}^{n-1}}(\pi) \otimes \id_{B_n \bar{B}_n}$ corresponds to a permutation in $\mS_n$ that acts on $(B\bar{B})_{1}^{n}$ and fixes the system $(B\bar{B})_n$. This leads to
	\begin{align*}
		\rho_{A\bar{A}(B\bar{B})_{1}^{n}} = \left( \id_{A\bar{A}} \otimes U_{(B\bar{B})_{1}^{n-1}}(\pi) \otimes \id_{B_n \bar{B}_n}\right)\left( \rho_{A\bar{A}(B\bar{B})_{1}^{n}} \right) \left( \id_{A\bar{A}} \otimes U_{(B\bar{B})_{1}^{n-1}}(\pi)^{*} \otimes \id_{B_n\bar{B}_n} \right) \, .
	\end{align*}
	Therefore, for any $\pi \in \mS_{n-1}$,
	\begin{align*}
		&\rho_{A\bar{A}(B\bar{B})_{1}^{n-1}B_n} \coloneqq \tr_{\bar{B}_n}(\rho_{A\bar{A}(B\bar{B})_{1}^{n}}) = \sum_{\substack{i,j \in [d_{A\bar{A}}] \\ u,v \in [d_{B_n}]\\ x\in [d_{\bar{B}_n}]}} \ketbra{i}{j} \otimes U_{(B\bar{B})_{1}^{n-1}}(\pi)\rho^{i,j,u,v,x,x}\otimes U_{(B\bar{B})_{1}^{n-1}}(\pi)^{*} \otimes \ketbra{u}{v} \\
		&= \left(\id_{A\bar{A}} \otimes  U_{(B\bar{B})_{1}^{n-1}}(\pi) \otimes \id_{B_n} \right) \left(\sum_{\substack{i,j \in [d_{A\bar{A}}] \\ u,v \in [d_{B_n}]\\ x\in [d_{\bar{B}_n}]}}  \ketbra{i}{j} \otimes \rho^{i,j,u,v,x,x} \otimes \ketbra{u}{v} \right) \left( \id_{A\bar{A}} \otimes U_{(B\bar{B})_{1}^{n-1}}(\pi)^{*} \otimes \id_{B_n}\right) \\
		&=\left(\id_{A\bar{A}} \otimes \otimes U_{(B\bar{B})_{1}^{n-1}}(\pi) \otimes \id_{B_n} \right) \left( \tr_{\bar{B}_n}(\rho_{A\bar{A}(B\bar{B})_{1}^{n}}\right) \left(\id_{A\bar{A}} \otimes U_{(B\bar{B})_{1}^{n-1}}(\pi)^{*} \otimes \id_{B_n} \right) \, .
	\end{align*}
	Hence,
	\begin{align*}
		\rho_{A\bar{A}(B\bar{B})_{1}^{n-1}B_n} \in \End^{\mS_{n-1}} \left(A\bar{A} \otimes (B\bar{B})^{\otimes (n-1)}\otimes B_n  \right) \, .
	\end{align*}
	Similarly, we also have
	\begin{align*}
		\rho_{A\bar{A}(B\bar{B})_{1}^{n-1}} \otimes\frac{\id_{B_n}}{d_{B}} \in \End^{\mS_{n-1}} \left(A\bar{A} \otimes (B\bar{B})^{\otimes (n-1)}\otimes B_n  \right) \, .
	\end{align*}
	This proves Claim~\eqref{eq:invariant2}. 
\end{proof}

%%%%%%

From Lemma~\ref{lem:invariant_matrices}, we will use the tools from Section~\ref{sec:Tools} to simplify Program~\ref{eq:sdp_hierarchy}. 
\par
%%%%%
Let $\Par(d_{\cH},n) \coloneqq \{ \lambda: \lambda \vdash_{d_{\cH}} n \}$ and $m_{\lambda}(\cH) \coloneqq |T_{\lambda,d_{\cH}}|$ for each $\lambda \in \Par(d_{\cH},n)$. From Lemma~\ref{lem:direct_groups_representative}, the representative matrix set for the action of $\mS_n$ on the space $\cA \otimes \cH^{\otimes n} \otimes \cB$ is described as follows,
\begin{align}
	\label{eq:representative_set}
	\{\id_{\cA} \otimes U_{\lambda} \otimes \id_{\cB} \}_{\lambda \in \Par(d_{\cH},n)} \, ,
\end{align}
where $\{U_{\lambda}\}_{\lambda \in \Par(d_{\cH},n)}$ are real matrices defined in Eq.~\eqref{eq:symmetric_representative_set}.
\par 
For each $\lambda \in \Par(d_{\cH},n)$, let $m_{\lambda}(\cA,\cH^n,\cB)$ be the number of rows of the matrix $(\id_{\cA} \otimes U_{\lambda} \otimes \id_{\cB} )$. From Lemma~\ref{lemma:psd_preserving} and the representative matrix set constructed in Eq.~\eqref{eq:representative_set}, we have the following proposition.
 %we have a bijective linear map $\Psi_{\cA,\cH^n,\cB}$ construct as follow
 \begin{proposition} The following map
 	\label{pro:bijective_map}
 	\begin{align} 
 		\begin{split}
 			\label{eq:bijective_map}
 			\Psi_{\cA,\cH^n,\cB}: \End^{\mS_n}(\cA \otimes \cH^{\otimes n} \otimes \cB) &\to \bigoplus_{\lambda \vdash_{d_{\cH}}n} \CC^{m_{\lambda}(\cA,\cH^n,\cB) \times m_{\lambda}(\cA,\cH^n,\cB)} \\
 			X &\mapsto \bigoplus_{\lambda \vdash_{d_{\cH}}n} \left(\id_{\cA} \otimes U_{\lambda}^{T} \otimes \id_{\cB} \right)X \left( \id_{\cA} \otimes U_{\lambda} \otimes \id_{\cB}\right) \, ,
 		\end{split}
 	\end{align}
   is a bijective linear map that preserve the positive semidefinitess property, i.e., $X \succeq 0$ if and only if $\Psi_{\cA,\cH^n,\cB}(X) \succeq 0$. 
 \end{proposition}

 Moreover, as a consequences of Eqs.~\eqref{eq:number_partitions} and~\eqref{eq:dim_END vs multiplicites}, we have the following proposition.
\begin{proposition}
	\label{pro:bounded_contraints}
	We have:
	\begin{align}
		&|\Par(d_{\cH},n)| \leq (n+1)^{d_{\cH}} \; \label{eq:number_partion},\\
		&m_{\lambda}(\cA,\cH^n,\cB) \leq d_{\cA}d_{\cB}(n+1)^{d_{\cH}(d_{\cH}-1)/2} \, : \; \forall \lambda \in \Par(d_{\cH},n) \; , \label{eq:number_size} \\
		&m(\cA,\cH^n,\cB) \coloneqq \dim \left[\End^{\mS_n}\left(\cA \otimes \cH^{\otimes n} \otimes \cB \right) \right] \leq d_{\cA}^2d_{\cB}^2(n+1)^{d_{\cH}^2} \label{eq: number_dimension} \; .
	\end{align}
\end{proposition}

\paragraph{A basis for the invariant subspace.} We can construct the canonical basis  of $\End^{\mathfrak{S}_n}\left(\cA \otimes \cH^{\otimes n} \otimes \cB \right)$  that consists of zero-one incidence matrices from the orbits of the group action of $\mS_n$ on the pairs $\left[d_{\cA} \times (d_{\cH})^n \times d_{\cB} \right]^2$(see~\cite{Klerk_07} or~\cite[Chapter 9]{Bachoc2012} for more information). In particular, let $i \in \left[d_{\cA} \times (d_{\cH})^n \times d_{\cB} \right] $ be the index of the standard basis of $\cA \otimes \cH^{\otimes n} \otimes \cB$. Then the orbit of the pair $(i,j) \in \left[d_{\cA} \times (d_{\cH})^n \times d_{\cB} \right]^2$ under the action of the group $\mS_n$ is given by
\begin{align*}
	O(i,j) = \{ (\pi(i), \pi(j) ): \pi \in \mathfrak{S}_n \},
\end{align*}   
where $\pi(i)$ is the index of the basis vector $\left( \id_{\cA} \otimes U_{\cH^n}(\pi)\otimes \id_{\cB}\right) \ket{i}$. The set $\left[d_{\cA} \times (d_{\cH})^n \times d_{\cB} \right]^2$ decomposes into orbits $O_{1},\dots,O_{m(\cA,\cH^n,\cB)}$ under the action of $\mathfrak{S}_n$, recall that $m(\cA,\cH^n,\cB)$ is the dimension of the invariant space $\End^{\mathfrak{S}_n}\left(\cA \otimes \cH^{\otimes n} \otimes \cB \right)$. Moreover, from Eq.~\eqref{eq: number_dimension}, one has $m(\cA,\cH^n,\cB) \leq d_{\cA}^2 d_{\cB}^2(n+1)^{d_\cH^2}$.
\par 

For each $r \in [m(\cA, \cH^n,\cB)]$, we construct a zero-one matrix $C_{r}$ of size $\left(d_{\cA} \times (d_{\cH})^n \times d_{\cB} \right)  \times \left(d_{\cA} \times (d_{\cH})^n \times d_{\cB} \right)$ given by
\begin{align}
	\label{eq:canonical_basis}
	(C_{r})_{ij} = 
	\begin{cases} 
		1 & \text{if }  (i,j) \in O_{r} \, , \\
		0 & \text{otherwise}.
	\end{cases}
\end{align}
The set $\{C_{1},\dots,C_{m(\cA,\cH^n,\cB)} \}$ forms a canonical basis of $\End^{\mathfrak{S}_n}\left(\cA \otimes \cH^{\otimes n} \otimes \cB \right)$~\cite[Chapter 9]{Bachoc2012}.
%\newpage
\par
Next, we will show how to utilize the symmetries in Program~\ref{eq:sdp_hierarchy} to simplify it. 
\par 
Let $\cH \coloneqq B\bar{B}$. For each $\cD \in \{A\bar{A} \otimes \cH^{\otimes n}, A \otimes \cH^{\otimes n}, A\bar{A} \otimes \cH^{\otimes n-1} \otimes B_{n} \}$, we set
\begin{align*}
	t(\cD) = \begin{cases}
		    n-1 \text{ \, \, if } \cD = A\bar{A} \otimes \cH^{\otimes n-1} \otimes B_{n}  \, ,\\
		    n \text{ \, \, \, \,  \, \,   otherwise. }
	\end{cases}
\end{align*}
 Let $\psi_{\cD}: \End^{\mS_{t(\cD)}} \left( \cD \right) \to \bigoplus_{ \lambda \in \Par(d_{\cH},t(\cD))} \CC^{m_{\lambda}(\cD) \times m_{\lambda}(\cD)}$ be the bijective linear map defined in Eq.~\eqref{eq:bijective_map}, where $d_{\cH} =  d_{B\bar{B}} = d_Bd_{\bar{B}}$. For any $X \in \End^{\mS_{t(\cD)}}\left(\cD \right)$  and  $\lambda \in \Par(d_{\cH},t(\cD))$, we write $\llbracket \psi_\cD(X) \rrbracket_{\lambda}$ for the block of $\psi_{\cD}(X)$ indexed by~$\lambda$. Let $m(\cD)$ denote the dimension of the invariant subspace $\End^{\mS_{t(\cD)}}\left(\cD \right)$, i.e., $m(\cD) \coloneqq  \dim  \End^{\mS_{t(\cD)}}\left(\cD \right)$,  and let $\{C_1(\cD),\dots,C_{m(\cD)}(\cD)\}$ denote the canonical basis of $\End^{\mS_{t(\cD)}}\left(\cD \right)$ defined in Eq.~\eqref{eq:canonical_basis}. Define $\cD_1 \coloneqq A\bar{A} \otimes \cH^{\otimes n}, \; \cD_2 \coloneqq A \otimes \cH^{\otimes n}$ and $\cD_3 \coloneqq  A\bar{A} \otimes \cH^{\otimes n-1} \otimes B_{n}$. Since $\rho_{A\bar{A}(B\bar{B})_{1}^{n}} \in \End^{\mS_n} \left(\cD_1 \right)$, we can write $\rho_{A\bar{A}(B\bar{B})_{1}^{n}} = \sum_{i=1}^{m(\cD_1)}x_i C_{i}(\cD_1)$ for some $x_i \in \CC, \, i  = 1,\dots,m(\cD_1)$. From Lemma~\ref{lem:invariant_matrices} and Proposition~\ref{pro:bijective_map}, we can construct a transformation $\Phi$, obtained by applying $\psi_{\cD_1}, \psi_{\cD_2}, \psi_{\cD_3}$ to the relevant matrices, which maps Program~\ref{eq:sdp_hierarchy} to an equivalent semidefinite program $\Phi(\SDP_{n}(\cN,M))$ as follows.
\begin{align}
\label{eq:pro:transform} \tag{$\Phi(\mathrm{SDP}_{n}(\cN,M))$}
\begin{split}
        & \mathrm{maximize: } \quad d_{\bar{A}}d_{B} \cdot \sum_{i=1}^{m(\cD_1)}x_i \cdot \tr \left[  \left(\cJ_{\bar{A}B_1}^{\cN}\otimes \Phi_{A\bar{B}_1} \right)C_{i}(\cD_1)_{A\bar{A}B_1\bar{B}_1}   \right] \\
		&\mathrm{s.t.}  \sum_{i=1}^{m(\cD_1)}x_i \cdot \left\llbracket \psi_{\cD_1}(C_i(\cD_1))\right\rrbracket_{\lambda} \succeq 0 \; , \forall  \lambda \in \mathrm{Par}(d_{B}d_{\bar{B}},n) \, , \\
		&\sum_{i=1}^{m(\cD_1)}x_i \cdot \tr \left[ C_i(\cD_1)\right] = 1 \, ,\\
		%& & \text{ For } \cD = A\bar{A} \otimes (B\bar{B})^{\otimes n}  \text{ and } \cD' = A\otimes (B\bar{B})^{\otimes n}: \\
		& d_{A} \cdot \sum_{i=1}^{m(\cD_1)}x_i \cdot \left\llbracket \psi_{\cD_2} \left( \tr_{\bar{A}}(C_i(\cD_1)) \right)  \right\rrbracket_{\lambda} \\
		&= \sum_{i=1}^{m(\cD_1)}x_i \cdot \left\llbracket  \psi_{\cD_2} \left( \id_A \otimes \tr_{A\bar{A}}(C_i(\cD_1)) \right)  \right\rrbracket_{\lambda} \, \forall  \lambda \in \mathrm{Par}(d_Bd_{\bar{B}},n) \, ,\\
		%& & \text{ For } \cD = A\bar{A} \otimes (B\bar{B})^{\otimes n}  \text{ and } \cD' = A\bar{A}\otimes (B\bar{B})^{\otimes n-1}\otimes B_n: \\
		 & d_{B} \cdot \sum_{i=1}^{m(\cD_1)}x_i \cdot \left\llbracket \psi_{\cD_3} \left( \tr_{\bar{B}_n}(C_i(\cD_1)) \right)  \right\rrbracket_{\lambda} \\
	  &= \sum_{i=1}^{m(\cD_1)}x_i \cdot \left\llbracket  \psi_{\cD_3} \left( \tr_{B_n\bar{B}_n}(C_i(\cD_1)) \otimes \id_{B} \right)  \right\rrbracket_{\lambda} \, \forall  \lambda \in \mathrm{Par}(d_Bd_{\bar{B}},n-1) \, , \\
	 & x_1,\dots,x_{m(\cD_1)} \in \CC \, .
   \end{split}
\end{align}

\hfill \break

\begin{theorem}
	\label{thm:convert_program}
	For $M,n \in \NN_{\geq 1}$. Let $\cN_{\bar{A}\to B}$ be a quantum channel. The semidefinite program $ \Phi (\SDP_{n}(\cN,M))$ has $\bigO\left(d_{A}^{2}d_{\bar{A}}^{2} (n+1)^{d_{B}^2d_{\bar{B}}^2}\right)$ variables and $\bigO \left((n+1)^{d_{B}d_{\bar{B}}} \right)$ positive semidefinite constraints involving matrices of size at most $\bigO \left(d_{A}d_{\bar{A}}d_{B}(n+1)^{d_{B}d_{\bar{B}}(d_{B}d_{\bar{B}}-1)/2} \right)$. As a consequence, the size of the Program $ \Phi (\SDP_{n}(\cN,M))$ is $\poly(d_{\bar{A}},n)$, for fixed $d_{A} = d_{\bar{B}} = M$  and $d_{B}$.  
\end{theorem}
\begin{proof}
    The proof of the theorem follows directly by using the results from Proposition~\ref{pro:bounded_contraints}. Firstly, the number of variables in Program~\ref{eq:pro:transform} is $m(\cD_1) = \dim \End^{\mS_n} \left( A\bar{A} \otimes (B\bar{B})^{\otimes n} \right)$. Thus, from Eq.~\eqref{eq: number_dimension}, it is bounded by $\bigO \left( d_{A}^{2}d_{\bar{A}}^{2} (n+1)^{d_{B}^2d_{\bar{B}}^2}\right)$. Since $|\Par(d_Bd_{\bar{B}},n)| \leq (n+1)^{d_{B}^2d_{\bar{B}}^2}$ according to Eq.~\eqref{eq:number_partion}, so the number of positive semidefinite (PSD) constraints is bounded by $\bigO \left( (n+1)^{d_{B}d_{\bar{B}}}\right)$. Finally, for every $q \in [3]$ and $\lambda \in \Par(d_Bd_{\bar{B}},n)$, by Eq.~\eqref{eq:number_size}, we have $m_{\lambda}(\cD_q) \leq d_{A}d_{\bar{A}}d_{B}(n+1)^{d_{B}d_{\bar{B}}(d_{B}d_{\bar{B}}-1)/2}$. This implies that the size of matrices involved in PSD constraints is at most $\bigO \left(d_{A}d_{\bar{A}}d_{B}(n+1)^{d_{B}d_{\bar{B}}(d_{B}d_{\bar{B}}-1)/2} \right)$. Therefore, when $d_{A} = d_{\bar{B}}= M$  and $d_{B}$ are fixed, the size of Program $ \Phi (\SDP_{n}(\cN,M))$ is bounded by $\poly(d_{\bar{A}},n)$, which completes the proof. 
\end{proof}

%\newpage

From Theorem~\ref{thm:convert_program}, we have that the size of the Program~\ref{eq:pro:transform} is bounded by a polynomial in $n$ and $d_{\bar{A}}$, therefore it can be solved in $\poly(d_{\bar{A}},n)$ time. However, directly applying the transformation $\Phi$ that maps Program~\ref{eq:sdp_hierarchy} into the equivalent Program~\ref{eq:pro:transform} requires exponential computations in $n$. Next, we will provide an efficient method to accomplish this task. In particular, let  $\{ C_1(\cD),\dots,C_{m(\cD)}(\cD)\}$ be the canonical basis of $\End^{\mS_{t(\cD)}}\left(\cD \right)$ defined in Eq.~\eqref{eq:canonical_basis}. Given $z_1,\dots,z_{m(\cD)} \in \CC$, we show that for any $\lambda \in \Par(d_{\cH},t(\cD))$, the block $ \left\llbracket \psi_{\cD}\left(\sum_{i=1}^{m(\cD)}z_i C_i(\cD) \right) \right \rrbracket_{\lambda}$ can be computed in $\poly(d_{\bar{A}},n)$ time, where $d_{\cH}$ is fixed. We then use this result to apply it to our specific spaces to reach the desired conclusion.  
\par 

Let $\{ C_1^{\cH}, \dots, C_{m(\cH)}^{\cH} \}$ denote the canonical basis of $\End^{\mS_n}\left( \cH^{\otimes n}\right)$ as defined in Eq.~\eqref{eq:canonical_basis}, where $m(\cH) \coloneqq \dim \End^{\mS_n}\left( \cH^{\otimes n}\right) $. The canonical basis of the space $\End^{\mS_n} \left( \cA \otimes \cH^{\otimes n} \otimes \cB \right)$ can be explicitly described as follows.
\begin{align}
	\{\ketbra{i_\cA}{j_\cA} \otimes C_t^{\cH} \otimes \ketbra{i_\cB}{j_\cB} \}_{\substack{i_\cA,j_\cA \in [d_\cA] \\ i_\cB,j_\cB \in [d_\cB]\\t \in [m(\cH)]}} \; .
\end{align}   
Let $\psi_{\cA,\cH,\cB}: \End^{\mS_n} \left( \cA \otimes \cH^{\otimes n} \otimes \cB \right) \to \bigoplus_{\lambda \in \Par(d_\cH,n)} \CC^{m_\lambda(\cA, \cH^n,\cB) \times m_\lambda(\cA, \cH^n,\cB)}$ be a bijective linear map defined in Eq.~\eqref{eq:bijective_map}. For any $\lambda \in \Par(d_\cH,n)$, we have

\begin{align}
\label{eq:compute_bolck}
	\begin{split}
	    \left \llbracket \psi_{\cA,\cH,\cB} \left(\ketbra{i_\cA}{j_\cA} \otimes C_t^{\cH} \otimes \ketbra{i_\cB}{j_\cB} \right) \right \rrbracket_{\lambda} &\coloneqq \left( \id_\cA \otimes U_{\lambda}^{T} \otimes \id_{\cB} \right)\left(\ketbra{i_\cA}{j_\cA} \otimes C_t^{\cH} \otimes \ketbra{i_\cB}{j_\cB} \right)  \left( \id_\cA \otimes U_\lambda \otimes \id_{\cB} \right) \\
	&= \ketbra{i_\cA}{j_\cA} \otimes U_{\lambda}^TC_t^{\cH}U_{\lambda} \otimes \ketbra{i_\cB}{j_\cB}. 
	\end{split}
\end{align}
We are going to show that $U_{\lambda}^TC_t^{\cH}U_{\lambda}$ can be computed in $\poly(n)$ time, for any $\lambda \in \Par(d_{\cH},n)$ and $t \in [m(\cH)]$. More generally, given $z_1,\dots, z_{m(\cH)} \in \CC$, for any $\lambda \in \Par(d_\cH,n)$, we will show that $\sum_{i=1}^{m(\cH)}z_i U_{\lambda}^T C_i^{\cH} U_{\lambda}$ can be computed in $\poly(n)$ time. In particular, from Eqs.~\eqref{eq:symmetric_representative_set} and~\eqref{eq:number_partitions}, we just need to show that we can compute $\sum_{i=1}^{m(\cH)}z_iu_{\tau}^{T}C_i^{\cH}u_{\gamma}$ in $\poly(n)$ time for any $\tau,\gamma \in T_{d_\cH,\lambda}$. We note that $u_{\tau}$, $u_{\gamma}$ and $C_{r}^{\cH}$ all have exponential size in $n$. As a direct consequence from~\cite[Lemma 2]{litjens2017semidefinite} (see also in~\cite[Lemma 4.5]{fawzi2022hierarchy}), we have the following lemma.    

\begin{lemma}
	\label{lem:computing_entries_block}
	Let $\lambda \in \Par(d_{\cH},n)$, and  $\tau,\gamma \in T_{\lambda,d_{\cH}}$. Given $z_1,\dots,z_{m(\cH)} \in \CC$. Then $\sum_{r=1}^{m(\cH)}z_r u_{\tau}^{T}C_{r}^{\cH}u_{\gamma}$ can be computed in polynomial time in $n$, for fixed $d_{\cH}$. As a direct consequence, for any $i_\cA,j_\cA \in [d_\cA]; i_\cB,j_\cB \in [d_\cB]; t \in [m(\cH)]$, we can determine the matrix $\left \llbracket \psi_{\cA,\cH,\cB} \left(\ketbra{i_\cA}{j_\cA} \otimes C_t^{\cH} \otimes \ketbra{i_\cB}{j_\cB} \right) \right \rrbracket_{\lambda}$ in $\poly(d_{\cA},d_{\cB},n)$ time for all $\lambda \in \Par(d_\cH,n)$.     
\end{lemma}
%%%
%The proof of Lemma~\ref{lem:computing_entries_block} is based on~\cite{litjens2017semidefinite} and~\cite[pp. 30-31]{Polak_thesis} (see also~\cite{fawzi2022hierarchy}), but we do it here for reader convenience.  We first introduce some notation.
\begin{proof}
	The proof of Lemma~\ref{lem:computing_entries_block} is based on~\cite{litjens2017semidefinite} and~\cite[pp. 30–31]{Polak_thesis} (see also~\cite{fawzi2022hierarchy}), and can be found in Appendix~\ref{appendix:computing_entries}.
\end{proof}

%%%

Consider the Program~\ref{eq:pro:transform}. For each $\cD \in \{A\bar{A} \otimes \cH^{\otimes n}, A \otimes \cH^{\otimes n }, A\bar{A} \otimes \cH^{\otimes n-1} \otimes B_{n} \}$ and any $X \in \End^{\mS_{t(\cD)}}(\cD)$, if all coordinates in the expansion of $X$ in the canonical basis $\End^{\mS_{t(\cD)}}(\cD)$ are given, we can compute the block $\llbracket \psi_\cD(X) \rrbracket_{\lambda}$ in $\poly(d_{\bar{A}}, n)$ time for all $\lambda \in \Par(d_{\cH}, t(\cD))$, as pointed out in Lemma~\ref{lem:computing_entries_block}. With this result, we demonstrate in the following theorem that the transformation $\Phi$ can be implemented in polynomial time with respect to $d_{\bar{A}}$ and $n$. As a conclusion, Program~\ref{eq:sdp_hierarchy} can be solved in $\poly(d_{\bar{A}},n)$ time.
\par
%%%%%%%%%%%%%%%%%%%%%%%%%%%%%%%%%%%%%%%
\begin{theorem}
	\label{thm:transformation}
	For $M,n\in \NN_{\ge 1}$. Let $\cN_{\bar{A}\to B}$ be a quantum channel. The transformation $\Phi$ that maps Program~\ref{eq:sdp_hierarchy} to Program~\ref{eq:pro:transform} can be done in $\poly(d_{\bar{A}}, n)$ time for fixed $d_{\cH} = d_{B}d_{\bar{B}}$.
\end{theorem}
Before proving Theorem~\ref{thm:transformation}, we recall some properties of the orbits generated by the natural action of the symmetric group, as presented in~\cite{gijswijt2009block}.

\textbf{Enumerating all orbits}. Let $\{C_1^{\cH},\dots,C_{m(\cH)}^{\cH} \}$ be the canonical basis of $\End^{\mS_n} \left(\cH^{\otimes n} \right)$ defined in Eq.~\eqref{eq:canonical_basis}. For each $r=  \{ 1,\dots,m(\cH) \}$, let $O_{r}^{\cH}$ be an orbit that corresponds to the matrix $C_{r}^{\cH}$,  we need to compute a representative element of $O_{r}^{\cH}$. In order to do so, we define a matrix $E^{(i,j)} \in \mathbb{Z}_{\geq 0}^{d_\cH \times d_\cH}$
\begin{align}
	\label{eq:number_pair}
	(E^{(i,j)})_{a,b} \coloneqq \left|\{v \in [n]: i_v = a,j_v = b \}\right| \; , \; \forall a,b \in[d_{\cH}] \, .
\end{align}  

Based on the construction given in  Eq.~\eqref{eq:number_pair}, for any two pairs $(i,j), (i',j')$ belonging to the sets $[d_\cH]^n \times [d_\cH]^n$, it holds that $(i',j')=(\pi(i), \pi(j))$ for some permutation $\pi \in \mathfrak{S}_n$ if and only if $E^{(i,j)} = E^{(i',j')}$. Therefore, there exists a direct one-to-one mapping between the orbits $\cbr{O_{r}^{\cH}}_{r\in[m(\cH)]}$ and matrices $E \in \mathbb{Z}_{\geq 0}^{d_\cH \times d_\cH}$ where the sum of all elements in $E$ is equal to $n$, i.e., $\sum_{a,b}E_{a,b} = n$. Thus, we can efficiently find a representative element for each $O_r^{\cH}$ in polynomial time by listing all non-negative integer solutions of the equation $\sum_{a,b \in [d_\cH]} E_{a,b} = n$. Moreover, let \( E^{i,j} \) be the matrix corresponding to the orbit \( O_r^{\cH} \). The size of \( O_r^{\cH} \) is given by
\begin{equation}
	\label{eq: size_of_orbit}
	|O_r^{\cH}| = \frac{n!}{\prod_{(a,b) \in [d_{\cH}] \times [d_{\cH}]} (E^{i,j})_{a,b}!} \, .
\end{equation}

\par 
For convenience, we will recall some notations. For each $\cD \in \{A\bar{A} \otimes \cH^{\otimes n}, A \otimes \cH^{\otimes n}, A\bar{A} \otimes \cH^{\otimes n-1} \otimes B_{n} \}$, we set $t(\cD) = n-1$ if $\cD = A\bar{A} \otimes \cH^{\otimes n-1} \otimes B_{n}$. Let $\{C_1(\cD),\dots,C_{m(\cD)}(\cD)\}$ denote the canonical basis of $\End^{\mS_{t(\cD)}}\left(\cD \right)$ defined in Eq.~\eqref{eq:canonical_basis}. Following the notation in the proof of Theorem~\ref{thm:convert_program}, let $\psi_{\cD}: \End^{\mS_{t(\cD)}} \left(\cD \right) \to \bigoplus_{ \lambda \in \Par(d_{\cH},t(\cD))} \CC^{m_{\lambda}(\cD) \times m_{\lambda}(\cD)}$ be the bijective linear map defined in Eq.~\eqref{eq:bijective_map}, where $m_{\lambda}(\cD)$ is the size of the block indexed by $\lambda$. Recall that, for any  $\lambda \in \Par(d_{\cH}, t(\cD))$, one has $ m_{\lambda}(\cD) \leq \poly(d_{\bar{A}},n)$ by Theorem~\ref{thm:convert_program}.

\begin{proof}[Proof of the Theorem~\ref{thm:transformation}]
    For $\cD_1 \coloneqq A\bar{A}\otimes \cH^{\otimes n}$ and note that $d_{\cH} = d_{B}d_{\bar{B}}$ is fixed. Let $\{C_1^{\cH},\dots,C_{m(\cH)}^{\cH}\}$ denote the canonical basis of $\End^{\mS_n} \left( \cH^{\otimes n}\right)$ defined in Eq.~\eqref{eq:canonical_basis} and let $\{O_1^{\cH},\dots,O_{m(\cH)}^{\cH}\}$ denote the set of orbits of pairs that corresponding to $\{C_1^{\cH},\dots,C_{m(\cH)}^{\cH}\}$. The following set of matrices
	\begin{align*}
		\mathcal{B}_1 = \{\ketbra{i}{j} \otimes \ketbra{x}{y} \otimes C_{t}^{\cH}\}_{\substack{i,j \in [d_A]\\ x,y \in [d_{\bar{A}}]\\ t \in [m(\cH)]}} 
	\end{align*}   
	is the canonical basis of $\End^{\mS_n} \left(\cD_1 \right)$. Since $\rho_{A\bar{A}(B\bar{B})_{1}^{n}} \in \End^{\mS_n}\left( \cD_1 \right)$, we can express it in terms of the canonical basis $\cB_1$ as
	 \begin{align*}
	 	\rho_{A\bar{A}(B\bar{B})_{1}^{n}} = \sum_{\substack{i,j \in [d_A]\\ x,y \in [d_{\bar{A}}]\\ t \in [m(\cH)]}}v_{i,j,x,y,t} \ketbra{i}{j} \otimes \ketbra{x}{y} \otimes C_{t}^{\cH} \, ,
	 \end{align*}  
	 where $v_{i,j,x,y,t} \in \CC$ for all $i,j \in [d_A]; x,y \in [d_{\bar{A}}]; t \in [m(\cH)]$, note that $m(\cH)$ is bounded by $\poly(n)$. We first write the Program $\Phi(\SDP_{n}(\cN,M)$ in terms of variables $v_{i,j,x,y,t}$ for $i,j \in [d_A]; x,y \in [d_{\bar{A}}]; t \in [m(\cH)]$.

%\newpage	 
	 
	  For simplicity of notation, for any $i,j \in [d_A], x,y \in [d_{\bar{A}}]$, and $t \in [m(\cH)]$, we denote $Z(i,j,x,y,t) = \ketbra{i}{j} \otimes \ketbra{x}{y} \otimes C_{t}^{\cH}$. Define  $\cD_2 \coloneqq A \otimes \cH^{\otimes n}$ and $\cD_3 \coloneqq  A\bar{A} \otimes \cH^{\otimes n-1} \otimes B_{n}$. 
	  
	  Program~\ref{eq:pro:transform} can be rewritten as
	\begin{IEEEeqnarray*}{rCl"r}
		& \, &  \mathrm{maximize } \quad d_{\bar{A}}d_{B} \cdot \sum_{\substack{i,j \in [d_A]\\ x,y \in [d_{\bar{A}}]\\ t \in [m(\cH)]}}v_{i,j,x,y,t} \cdot \tr \left[  \left(\cJ_{\bar{A}B_1}^{\cN}\otimes \Phi_{A\bar{B}_1} \right)  \tr_{(B\bar{B})_{2}^{n}}\left( Z(i,j,x,y,t)\right)   \right] &\\
		\mathrm{ s.t.}&  &\sum_{\substack{i,j \in [d_A]\\ x,y \in [d_{\bar{A}}]\\ t \in [m(\cH)]}}v_{i,j,x,y,t} \left\llbracket \psi_{\cD_1} \left( Z(i,j,x,y,t)\right)\right\rrbracket_{\lambda} \succeq 0 \; , \forall  \lambda \in \Par(d_{B}d_{\bar{B}},n)\, , \;  \\
		& & \sum_{\substack{i,j \in [d_A]\\ x,y \in [d_{\bar{A}}]\\ t \in [m(\cH)]}}v_{i,j,x,y,t} \cdot \tr \left(Z(i,j,x,y,t) \right)  = 1 \, , \\
		%& & \text{ Let } \cD_2 \coloneqq A\otimes (B\bar{B})^{\otimes n}. \text{ For each } \lambda \in \Par(d_Bd_{\bar{B}},n): \\
		& & d_{A} \cdot \sum_{\substack{i,j \in [d_A]\\ x,y \in [d_{\bar{A}}]\\ t \in [m(\cH)]}}v_{i,j,x,y,t} \left\llbracket \psi_{\cD_2} \left( \tr_{\bar{A}} \left(Z(i,j,x,y,t) \right) \right)  \right\rrbracket_{\lambda} \\ 
		& &= \sum_{\substack{i,j \in [d_A]\\ x,y \in [d_{\bar{A}}]\\ t \in [m(\cH)]}}v_{i,j,x,y,t}  \left\llbracket  \psi_{\cD_2} \left( \id_A \otimes \tr_{A\bar{A}} \left( Z(i,j,x,y,t) \right) \right)  \right\rrbracket_{\lambda} \, \forall \lambda \in \Par(d_Bd_{\bar{B}},n) \, , \\
		%& & \text{ Let } \cD_3 \coloneqq A\bar{A}\otimes (B\bar{B})^{\otimes n-1}\otimes B_n. \text{ For each } \lambda \in \Par(d_Bd_{\bar{B}},n-1): \\
		& & d_{B} \cdot \sum_{\substack{i,j \in [d_A]\\ x,y \in [d_{\bar{A}}]\\ t \in [m(\cH)]}}v_{i,j,x,y,t} \left\llbracket \psi_{\cD_3} \left( \tr_{\bar{B}_n}\left( Z(i,j,x,y,t)\right) \right)  \right\rrbracket_{\lambda} \\ 
		& &= \sum_{\substack{i,j \in [d_A]\\ x,y \in [d_{\bar{A}}]\\ t \in [m(\cH)]}}v_{i,j,x,y,t}  \left\llbracket  \psi_{\cD_3} \left( \tr_{B_n\bar{B}_n}\left( Z(i,j,x,y,t)\right) \otimes \id_{B} \right)  \right\rrbracket_{\lambda} \, \forall \lambda \in \Par(d_Bd_{\bar{B}},n-1) \, , \\ \\
		& & v_{i,j,x,y,t} \in \CC \; \text{ for all } i,j \in [d_A]; x,y \in[d_{\bar{A}}], \text{ and } t\in [m(\cH)]  \, \, .
	\end{IEEEeqnarray*}

We will demonstrate that the objective function and constraints in Program~\ref{eq:pro:transform} can be determined explicitly in $\poly(d_{\bar{A}},n)$ time.
\par 

Firstly, for any $t \in [m(\cH)]$, let $(a,b)$ be a representative element of $O_{t}^{\cH}$. We can compute the trace of each matrix in the canonical basis $\cB_1$ of $\End^{\mS_n}(\cD_1)$ efficiently. Indeed,
	 \begin{align*}
	 	\tr \left( \ketbra{i}{j} \otimes \ketbra{x}{y} \otimes C_{t}^{\cH} \right) = \begin{cases}
	 	    |O_{t}^{\cH}|, & \text{if $i=j, x=y$ and $a=b$} \, ,\\
	 		0, & \text{otherwise.}
	 	\end{cases}
	 \end{align*} 
Therefore, we can compute $\tr (Z(i,j,x,y,t))$ efficiently for all $i,j \in [d_A]; x,y \in [d_{\bar{A}}]; t \in [m(\cH)]$. In addition, from Lemma~\ref{lem:computing_entries_block}, we can determine $\left\llbracket \psi_{\cD_1} \left( Z(i,j,x,y,t)\right)\right\rrbracket_{\lambda}$ in $\poly(d_{\bar{A}},n)$ time for all $\lambda \in \Par(d_{B}d_{\bar{B}},n)$. 

\par 
Next, for each $q \in [3]$, given $z_1, \dots, z_{m(\cD_q)} \in \CC$ and consider a matrix $X = \sum_{i=1}^{m(\cD_q)}z_iC_i(\cD_q)$. From Lemma~\ref{lem:computing_entries_block}, we can determine all blocks in the diagonal matrix $\psi_{\cD_q}(X)$ in $\poly(d_{\bar{A}},n)$ time. Using this observation, to show $\left\llbracket \psi_{\cD_2} \left( \tr_{\bar{A}} \left(Z(i,j,x,y,t) \right) \right)  \right\rrbracket_{\lambda},  \left\llbracket  \psi_{\cD_2} \left( \id_A \otimes \tr_{A\bar{A}} \left( Z(i,j,x,y,t) \right) \right)  \right\rrbracket_{\lambda}$ for $\lambda \in \Par(d_{B}d_{\bar{B}},n)$ and $\left\llbracket \psi_{\cD_3} \left( \tr_{\bar{B}_n}\left( Z(i,j,x,y,t)\right) \right)  \right\rrbracket_{\lambda}, \left\llbracket  \psi_{\cD_3} \left( \tr_{B_n\bar{B}_n}\left( Z(i,j,x,y,t)\right) \otimes \id_{B} \right)  \right\rrbracket_{\lambda}$ for $\lambda \in \Par(d_{B}d_{\bar{B}},n-1)$ can be computed in $\poly(d_{\bar{A}},n)$ time. It suffices to show  that, for all $i,j \in [d_A] \, ; x,y  \in [d_{\bar{A}}] \, ; t \in [m(\cH)]$, expanding $\tr_{A} \left( \ketbra{i}{j} \otimes \ketbra{x}{y} \otimes C_{t}^{\cH} \right)$, $\id_A \otimes \tr_{A\bar{A}}\left( \ketbra{i}{j} \otimes \ketbra{x}{y} \otimes C_{t}^{\cH} \right)$ in the canonical basis of $\End^{\mS_n}(\cD_2)$ and $\tr_{\bar{B}_n} \left( \ketbra{i}{j} \otimes \ketbra{x}{y} \otimes C_{t}^{\cH} \right)$; $\tr_{B_n\bar{B}_n}\left( \ketbra{i}{j} \otimes \ketbra{x}{y} \otimes C_{t}^{\cH} \right)$ in the canonical basis $\End^{\mS_{n-1}}(\cD_3)$ can be done in $\poly(d_{\bar{A}},n)$ time.
\par 
 We have    
	\begin{align*}
		\tr_{\bar{A}}\left(\ketbra{i}{j} \otimes \ketbra{x}{y} \otimes C_{t}^{\cH} \right) = \begin{cases}
			\ketbra{i}{j} \otimes C_{t}^{\cH}, & \text{if $x=y$ }\\
			0, & \text{otherwise.}
		\end{cases}
	\end{align*} 
	And, 
	\begin{align*}
		\tr_{A\bar{A}} \left(\ketbra{i}{j} \otimes \ketbra{x}{y} \otimes C_{t}^{\cH} \right) = \begin{cases}
			C_{t}^{\cH}, & \text{if $i=j$ and $x=y$}\\
			0, & \text{otherwise.}
		\end{cases}  
	\end{align*}
On the other hand, the canonical basis of the space $\End^{\mS_n}\left( \cD_2 \right)$ is described as follows,
\begin{align*}
	\mathcal{B}_2 = \{\ketbra{i}{j} \otimes C_{t}^{\cH}\}_{\substack{i,j \in [d_A] \\ t \in [m(\cH)]}} \; .
\end{align*}   
Therefore, for any $i,j \in [d_A]$; $x,y \in [d_{\bar{A}}]$, and $t \in [m(\cH)]$, we can efficiently compute coordinates of the matrices $\tr_{\bar{A}} \left(\ketbra{i}{j} \otimes \ketbra{x}{y} \otimes C_{t}^{\cH} \right)$ and $\id_{A} \otimes \tr_{A\bar{A}}\left(\ketbra{i}{j} \otimes \ketbra{x}{y} \otimes C_{t}^{\cH} \right)$ over the canonical basis $\cB_2$.

Let $\ell(\cH) \coloneqq \dim \End^{\mS_{n-1}}\left(\cH^{\otimes n-1} \right)$. Note that, from Eq.~\eqref{eq: number_dimension}, $\ell(\cH)$ is bounded by $\poly(n)$. Consider the invariant subspace $\End^{\mS_{n-1}} \left(\cD_3 \right)$, let $\{K_1^{\cH},\dots, K_{\ell(\cH)}^{\cH} \}$ be the canonical basis of the space $\End^{\mS_{n-1}}\left(\cH^{\otimes n-1} \right)$ defined in Eq.~\eqref{eq:canonical_basis}. The following set of matrices
	\begin{align*}
		\mathcal{B}_3 = \{\ketbra{i}{j} \otimes \ketbra{x}{y} \otimes K_{t}^{\cH} \otimes \ketbra{u}{v} \}_{\substack{i,j \in [d_A] \\x,y \in [d_{\bar{A}}] \\ t \in [\ell(\cH)] \\ u,v \in [d_{B_n}]}} 
	\end{align*}   
	is the canonical basis of the space $\End^{\mS_{n-1}} \left(\cD_3 \right)$.

	  We will show that all coordinates of $\tr_{\bar{B}_n} \left(\ketbra{i}{j} \otimes \ketbra{x}{y} \otimes C_{t}^{\cH} \right)$ over $\cB_3$ can be computed in $\poly(d_{\bar{A}},n)$ time for any $i,j \in [d_A] ; x,y \in [d_{\bar{A}}]$, and $t \in [m(\cH)]$. Firstly, for any $t \in [m(\cH)]$, we can write
	\begin{align*}
		C_{t}^{\cH} = \sum_{(a,b) \in O_t^{\cH}} \ketbra{a_1}{b_1} \otimes \ketbra{a_2}{b_2} \otimes \dots \otimes \ketbra{a_n}{b_n} \, .
	\end{align*}
	Let $(a,b) \in [d_{B\bar{B}}]^n \times [d_{B\bar{B}}]^n$ be a representative element of $O_t^{\cH}$ following the construction in Eq. \eqref{eq:number_pair}. We denote by $S$ the set of all distinct pairs in  $ \{ (a_1,b_1),(a_2,b_2),\dots,(a_n,b_n) \}$, i.e., $S = \{(a_1,b_2) \cup (a_2,b_2) \cup \dots \cup (a_n,b_n) \}$. Under the action of $\mS_{n-1}$, the set $O_{t}^{\cH}$ can be partitioned into $|S|$ sets $\{O_{(c,d)}^{t}\}_{(c,d) \in S}$, where $O_{(c,d)}^{t}$ is defined as $\{(u,v) \in O_{t}^{\cH}: u_n = c, v_n = d \}$, noting that the size of $O_{(c,d)}^{t}$ can be computed in $\poly(n)$ time. For any $c \in [d_{B \bar{B}}]$, using the natural identification between $[d_{B\bar{B}}]$ and $[d_B] \times [d_{\bar{B}}]$ ($[d_{B\bar{B}}] \cong [d_B] \times [d_{\bar{B}}]$), and write $c = (c^{B},c^{\bar{B}}) \in [d_B] \times [d_{\bar{B}}]$. We have
	\begin{align*}
		C_{t}^{\cH} &= \sum_{(c,d) \in S} \sum_{(u,v) \in O_{(c,d)}^t} \ketbra{u_1}{v_1} \otimes \dots \otimes \ketbra{u_n = c}{v_n = d} \\
		&= \sum_{(c,d) \in S} \sum_{(u,v) \in O_{(c,d)}^t} \ketbra{u_1}{v_1} \otimes \dots \otimes \ketbra{c^{B_n}}{d^{B_n}}\otimes\ketbra{c^{\bar{B}_n}}{d^{\bar{B}_n}} \, . 
	\end{align*}  
 Therefore,
	\begin{align*}
		&\tr_{\bar{B}_n}(\ketbra{i}{j} \otimes \ketbra{x}{y} \otimes C_{t}^{\cH})\\
		&= \tr_{\bar{B}_n} \left( \ketbra{i}{j} \otimes \ketbra{x}{y} \otimes \sum_{(c,d) \in S} \sum_{(u,v) \in O_{(c,d)}^t} \ketbra{u_1}{v_1} \otimes \dots \otimes \ketbra{c^{B_n}}{d^{B_n}}\otimes\ketbra{c^{\bar{B}_n}}{d^{\bar{B}_n}} \right) \\
		&= \sum_{\substack{(c,d) \in S \\ c^{\bar{B}_n}=d^{\bar{B}_n}}} \sum_{(u,v) \in O_{(c,d)}^t} \ketbra{i}{j} \otimes \ketbra{x}{y} \otimes \ketbra{u_1}{v_1} \otimes \dots \otimes \ketbra{u_{n-1}}{v_{n-1}}\otimes \ketbra{c^{B_n}}{d^{B_n}}
	\end{align*}
        In addition, we can determine an index $t' \in [\ell(\cH)]$ such that 
	\begin{align*}
		K_{t'} = \sum_{\substack{(u,v) \in O_{(c,d)}^t \\ c^{\bar{B}_n}=d^{\bar{B}_n}}} \ketbra{u_1}{v_1} \otimes \dots \otimes \ketbra{u_{n-1}}{v_{n-1}} \, 
	\end{align*} 
	in $\poly(n)$ time by considering all representative elements corresponding to $K_1^{\cH},\dots, K_{\ell(\cH)}^{\cH} $. Thus, the coordinates of $\tr_{\bar{B}_n}\left( \ketbra{i}{j} \otimes \ketbra{x}{y} \otimes C_{t}^{\cH} \right)$ over the canonical basis $\cB_3$ can be computed in $\poly(n)$ time. The same method can be used for $\tr_{B_n\bar{B}_n} \left( \ketbra{i}{j} \otimes \ketbra{x}{y} \otimes C_{t}^{\cH} \right) \otimes \id_{B}$.

Finally, we will show that the objective function can be explicitly described by using variables $v_{i,j,x,y,t}$ for $i,j \in [d_A] \, ; x,y \in [d_{\bar{A}}]$, and $t \in [m(\cH)]$, in $\poly(d_{\bar{A}},n)$ time. In particular, the expression $\rho_{A\bar{A}B_1\bar{B}_1}$ can be described using variables $v_{i,j,x,y,t}$ for $i,j \in [d_A] \, ; x,y \in [d_{\bar{A}}]$, and $t \in [m(\cH)]$, in $\poly(d_{\bar{A}},n)$ time. Similarly to the previous argument, we need to show for any $i,j \in [d_A] \, ; x,y \in [d_{\bar{A}}]$, and $t \in [m(\cH)]$, the value $\tr_{(B\bar{B})_2^n}\left( \ketbra{i}{j} \otimes \ketbra{x}{y} \otimes C_{t}^{\cH} \right)$ can be computed efficiently.

Firstly, recall that for $x, y \in [m]^n$, their Hamming distance $H(x,y)$ is defined as $H(x,y) \coloneqq |\{i\in [n]: x_i \neq y_i \}|$. Let $(a,b)$ be a representative element of $O_{t}^{\cH}$. Let $C_{0}(t)$ be the size of $O_{t}^{\cH}$ and  $C_1(t)$ be the size of set $\{(u,v) \in O_t^{\cH}: u_1 = a_1,v_1 = b_1\}$. We have
\begin{align*}
	\tr_{(B\bar{B})_{2}^{n}} \left(\ketbra{i}{j} \otimes \ketbra{x}{y} \otimes C_{t}^{\cH} \right) &= \tr_{(B\bar{B})_{2}^{n}} \left(  \sum_{(c,d) \in O_t^{\cH}} \ketbra{i}{j} \otimes \ketbra{x}{y} \otimes \ketbra{c_1}{d_1} \otimes \ketbra{c_2}{d_2} \otimes \dots \otimes \ketbra{c_n}{d_n} \right) \\
	& = \begin{cases}
		C_0(t)\left( \ketbra{i}{j} \otimes \ketbra{x}{y} \otimes \ketbra{a_1}{a_1} \right)   \text{ if } H(a,b) = 0 \, ,\\
		C_1(t) \left( \ketbra{i}{j} \otimes \ketbra{x}{y} \otimes \ketbra{a_1}{b_1} \right) \text{ if } H(a,b) = 1 \text{ and } a_1 \neq b_1 \, ,\\
		0 \; \text{ otherwise. } 
	\end{cases}
\end{align*}         
Here, we used the fact that the Hamming distance is invariant under the action of the symmetric group; i.e., $H(x, y) = H(\pi \cdot x, \pi \cdot y)$ for all $\pi \in \mathfrak{S}_n$, where $\pi \cdot (x_1, \dots, x_n) = (x_{\pi^{-1}(1)}, \dots, x_{\pi^{-1}(n)})$. In addition, following Eq.~\eqref{eq: size_of_orbit}, $C_0(t)$ and $C_1(t)$ can be determined in $\poly(n)$ time, which completes the proof.
\end{proof}

%%%%%%%
\paragraph*{Acknowledgements.} We would like to thank Sven Polak for his useful comments and help in the proof of Lemma~\ref{lem:direct_groups_representative}. HT would like to thank Omar Fawzi and Mario Berta for the useful discussions.
 
%%%%%%%%%%%%%%%%%%%%%%%%%%%%%%%%%%%%%%%%%%%%%%	

%\section*{Acknowledgments} We would like to thank

%%%%%%%%%%%%%%%%%%%%%%%%%%%%%%%%%%%%%%%%%%%%%%%%%%%%%%%% 
%\bibliographystyle{alpha}
%\bibliographystyle{arxiv_no_month}
%\bibliography{bibliofile}

\newpage
\phantomsection
\addcontentsline{toc}{section}{Bibliography}
\bibliographystyle{alphaurl}
\bibliography{bibliofile}
\newpage
%%%%%%%%%%%%%%%%%%%%%%%%%%%%%%%%%%%%%%%%%%%%%%%%%%%%%%%%

\appendix
 \section{Proof of Lemma~\ref{lem:direct_groups_representative}}
 \label{Apen:representative_set}
\begin{proof}
	Let $\{U_{1}^{(1)},\dots,U_{k_1}^{(1)}\}$ and $\{U_{1}^{(2)},\dots,U_{k_2}^{(2)} \}$ be the representative matrix sets that correspond to the action of $G_1$ and $G_2$ on $\cH_1$ and $\cH_2$, respectively.  That is, let $\cH_1 = \oplus_{i=1}^{k_1} \oplus_{j=1}^{m_i} \cH_{i,j}^{(1)}$ and  $\cH_2 = \oplus_{i=1}^{k_2} \oplus_{j=1}^{s_i} \cH_{i,j}^{(2)}$ be decompositions of $\cH_1$ and $\cH_2$ into irreducible modules, such that $\cH_{i,j}^{(1)} \cong \cH_{i',j'}^{(1)}$ if and only if~$i=i'$ (for $i \in [k_1]$ and~$j \in [m_i]$), and $\cH_{i,j}^{(2)} \cong \cH_{i',j'}^{(2)}$ if and only if~$i=i'$ (for $i \in [k_2]$ and~$j \in [s_i]$). For each $i \in [k_1]$ and~$j \in [m_i]$, let~$u^{(1)}_{i,j} \in \cH_{i,j}^{(1)}$ be a nonzero vector such that for each~$i \in [k_1]$ and~$j,j' \in [m_i]$ there is a bijective $G_1$-equivariant linear map $\phi^{(1)}_{i,j,j'}$ from $\cH_{i,j}^{(1)}$ to $\cH_{i,j'}^{(1)}$ mapping $u^{(1)}_{i,j}$ to $u^{(1)}_{i,j'}$.  Assume that the matrix~$U_{i}^{(1)}:=(u_{i,j}^{(1)} \, | j \in [m_i])$ is the $i$-th matrix in the representative matrix set for the action of~$G_1$ on~$\cH_1$, for~$i \in [k_1]$.
	
	Similarly,  For each $i \in [k_2]$ and~$j \in [s_i]$, let~$u^{(2)}_{i,j} \in \cH_{i,j}^{(2)}$ be a nonzero vector such that for each~$i \in [k_2]$ and~$j,j' \in [s_i]$ there is a bijective $G_2$-equivariant linear map $\phi^{(2)}_{i,j,j''}$ from $\cH_{i,j}^{(2)}$ to $\cH_{i,j'}^{(2)}$ mapping $u^{(2)}_{i,j}$ to $u^{(2)}_{i,j'}$.  Assume that the matrix~$U_{i}^{(2)}:=(u_{i,j}^{(2)} \, | j \in [s_i])$ is the $i$-th matrix in the representative matrix set for the action of~$G_2$ on~$\cH_2$, for~$i \in [k_2]$.
	\par
	According to~\cite[Theorem 10]{rep_1977}, the decomposition of~$\cH_1 \otimes \cH_2$ into irreducible representations for~$G_1 \times G_2$ is $\oplus_{i=1}^{k_1} \oplus_{i'=1}^{k_2} \oplus_{j=1}^{m_i} \oplus_{j'=1}^{s_i} \cH_{i,j}^{(1)} \otimes \cH_{i',j'}^{(2)}$, and $\cH_{i,j}^{(1)} \otimes \cH_{i',j'}^{(2)} \cong \cH_{i'' ,j''}^{(1)} \otimes \cH_{i''',j'''}^{(2)}$ if and only if~$i=i''$ and~$i'=i'''$.  Now, if~$u_{i,j}^{(1)} \otimes u_{i',j'}^{(2)} \in \cH_{i,j}^{(1)} \otimes \cH_{i',j'}^{(2)}$ and~$u_{i,j''}^{(1)} \otimes u_{i',j'''}^{(2)} \in \cH_{i,j''}^{(1)} \otimes \cH_{i',j'''}^{(2)}$, the map $\phi^{(1)}_{i,j,j'} \otimes \phi^{(2)}_{i',j'',j'''}$ is a $G_1 \times G_2$-equivariant map from  $\cH_{i,j}^{(1)} \otimes \cH_{i',j'}^{(2)}$ to  $\cH_{i,j''}^{(1)} \otimes \cH_{i',j'''}^{(2)}$ mapping $u_{i,j}^{(1)} \otimes u_{i',j'}^{(2)} $ to $u_{i,j''}^{(1)} \otimes u_{i',j'''}^{(2)} $. It follows that the representative vectors for the action of~$G_1 \times G_2$ on $\cH_1 \otimes \cH_2$ are the vectors $u_{i,j}^{(1)} \otimes u_{i',j'}^{(2)} $, and that the representative matrix set is $\{ U_{i}^{(1)} \otimes U_{j}^{(2)} : i=1,\dots,k_1, \, j = 1,\dots,k_2 \} $, as desired.
\end{proof}
 
 \section{Proof of Lemma~\ref{lem:computing_entries_block}}
 \label{appendix:computing_entries}
 The proof of Lemma~\ref{lem:computing_entries_block} is based on~\cite{litjens2017semidefinite} and~\cite[pp. 30-31]{Polak_thesis} (see also ~\cite{fawzi2022hierarchy}), but we provide it here for the reader's convenience. We first introduce some notation and basic results. 
 
 %\textbf{Polynomial on a vector space.} For a finite dimensional complex vector space $\cH$, the \emph{dual vector space} $\cH^{*}$ of $\cH$ is the vector space of all linear transformations $\varphi: \cH \to \CC$. The \emph{coordinate ring} of $\cH$, denoted $\mathcal{O}(\cH)$, is the algebra consisting of all $\CC$-linear combinations of products of elements from $\cH^{*}$. An element of $\mathcal{O}(\cH)$ is called a polynomial on $\cH$. A polynomial $p \in \mathcal{O}(\cH)$ is called \emph{homogeneous} if it is a $\CC$-linear combination of a product of $n$ non-constant elements of $\cH^{*}$ (for a fixed non-negative integer $k$). We denote by $\mathcal{O}_{n}(\cH)$ the set all homogeneous polynomials of degree $n$.
 
 For a finite dimensional complex vector space $\mathcal{H}$, the \emph{dual vector space} $\mathcal{H}^{*}$ of $\mathcal{H}$ is the vector space of all linear transformations $\varphi: \mathcal{H} \to \mathbb{C}$. The \emph{coordinate ring} of $\mathcal{H}$, denoted as $\mathcal{O}(\mathcal{H})$, is the algebra consisting of all $\mathbb{C}$-linear combinations of products of elements from $\mathcal{H}^{*}$. An element of $\mathcal{O}(\mathcal{H})$ is called a polynomial on $\mathcal{H}$. A polynomial $p \in \mathcal{O}(\mathcal{H})$ is called homogeneous if it is a $\mathbb{C}$-linear combination of a product of $n$ non-constant elements of $\mathcal{H}^{*}$ (for a fixed non-negative integer $n$). We denote by $\mathcal{O}_{n}(\mathcal{H})$ the set of all homogeneous polynomials of degree $n$.
 
 \par 
 
 Set $W_{\cH} \coloneqq \cH \otimes \cH$, for each $p \coloneqq (x,y) \in [d_{\cH}] \times [d_{\cH}]$, define
 \begin{align*}
 	a_{p} \coloneqq x \otimes y \in W_{\cH} \, ,
 \end{align*}
 then the set $\cW \coloneqq \{a_{p}: p \in [d_{\cH}] \times [d_{\cH}] \}$ is a basis of $W_{\cH}$. Let $\cW^{*} \coloneqq \{a_{p}^{*}: p \in [d_\cH] \times [d_\cH] \}$ be the corresponding dual basis for $W_{\cH}^{*}$.
 
 Using the natural identification of $([d_{\cH}]\times [d_{\cH}])^n$ and $\left[(d_{\cH})^n \right]^2$, for any $r \in [m(\cH)]$, we have 
 \begin{align*}
 	C_{r}^{\cH} \coloneqq \sum_{(p_1,\dots,p_n) \in O_{r}^{\cH}} a_{p_1} \otimes \dots \otimes a_{p_n} \in W_{\cH}^{\otimes n} \, .
 \end{align*}
 Note that $C_{r}^{\cH}$ can be obtained from $\mathrm{vec}\left(C_{r}^{\cH} \right)$ by applying the permutation operator which maps $\left(\cH^{\otimes n} \right)^{\otimes 2}$ to $\left( \cH^{\otimes 2} \right)^{\otimes n}$. For every $(p_1,\dots,p_n) \in [(d_{\cH})^2]^n$, let 
 \begin{align}
 	\mu(p_1,\dots,p_n) \coloneqq a_{p_1}^{*} \cdots a_{p_n}^{*} \in \mathcal{O}_{n}(W_\cH)
 \end{align}
 be a degree $n$ monomial expressed in the basis $\cW^{*}$. Note that, for a fixed $r \in [m(\cH)]$, $\mu(p_1,\dots,p_n)$ is the same monomial, for every $(p_1,\dots,p_n) \in O_{r}^{\cH}$. We denote this monomial by $\mu \left( O_{r}^{\cH} \right)$. Moreover, $\{O_{r}^{\cH} \}_{r \in [m(\cH)]}$ partitions $\left[(d_\cH)^2 \right]^n$ into disjoint subsets. Therefore, there exists a bijection between $\{O_{r}^{\cH}\}_{r \in [m(\cH)]}$ and the set of degree $n$ monomials expressed in the basis $\cW^{*}$.
 
 %We now recall the relation between $\mathfrak{S}_n$-orbits on $([d_{\cH}] \times [d_{\cH}])^n$ and monomials of degree $n$ expressed in the dual basis $W_{\cH}^{*}$ of $W_{\cH}$. For each $O_{r}^{\cH}$ with $r \in [m^{\cH}]$, the monomial $\mu(O_{r}^{\cH})$ is defined as
 %\begin{align*}
 %	\mu(O_{r}^{\cH}) \coloneqq a_{p_1}^{*}\cdots a_{p_n}^{*} \in \mathcal{O}_{n}(W_{\cH}) \, , 
 %\end{align*}
 %where $(p_1,\dots,p_n)$ is an arbitrary element of $O_{r}^{\cH}$. By using the natural identification of $([d_{\cH}]\times [d_{\cH}])^n$ and $\left[(d_{\cH})^n \right]^2$ and note that the monomial does not depend on the choice of $(p_1,\dots,p_n)$. 
 
 Let $\zeta: (W_{\cH}^{*})^{\otimes n} \to \mathcal{O}_{n}(W_{\cH})$ be a linear function defined by
 \begin{align*}
 	\zeta(w_1^{*}\otimes \dots \otimes w_{n}^{*}) \coloneqq w_1^{*} \cdots w_{n}^{*}  \text{ for all } w_1^{*},\dots,w_{n}^{*} \in W_{\cH}^{*} \, .
 \end{align*}
 %We give a bijection between $\{O_{1}^{\cH},\dots,O_{m(\cH)}^{\cH}\}$ and the set of degree $n$ monomials expressed in the basis of $W_{\cH}^{*}$ as follow. 
 
 %Let $\zeta: (W_{\cH}^{*})^{\otimes n} \to \mathcal{O}_{n}(W_{\cH})$ be a linear function, which is defined as
 %\begin{align*}
 %	\zeta(w_1^{*}\otimes \dots \otimes w_{n}^{*}) \coloneqq w_1^{*} \cdots w_{n}^{*}  \text{ for all } w_1^{*},\dots,w_{n}^{*} \in W_{\cH}^{*} \, .
 %\end{align*}
 We denote $\overline{w} = \zeta(w)$ for all $w \in (W_{\cH}^{*})^{\otimes n}$. For any $\lambda \vdash_{d_\cH} n$ and $\tau,\gamma \in T_{\lambda,d_{\cH}}$, define the polynomial $G_{\tau,\gamma} \in \CC[x_{i,j}: i,j = 1,\dots,d_{\cH}]$ by 
 \begin{align}
 	\label{eq:polynomial}
 	G_{\tau,\gamma}(X) \coloneqq \sum_{\substack{\tau' \sim \tau \\ \gamma' \sim \gamma }} \sum_{c,c' \in C_{\lambda}} \mathrm{sgn}(cc') \prod_{y \in Y(\lambda)} x_{\tau'c(y),\gamma'c'(y)} \, ,
 \end{align} 
 for $X = (x_{i,j})_{i,j = 1}^{d_{\cH}} \in \CC^{d_{\cH} \times d_{\cH}}$. Ref~\cite[Proposition 3]{litjens2017semidefinite} and \cite[Theorem 7]{gijswijt2009block} show that the polynomial in Eq~\eqref{eq:polynomial} can be computed in polynomial time, i.e., expressed as a linear combination of monomials in variables $x_{i,j}$ in $\poly(n)$ time.
 \begin{lemma} [\cite{gijswijt2009block,litjens2017semidefinite}]
 	\label{lem:polynomial}
 	Let $\lambda \in \Par(d_{\cH},n)$ and every $\tau,\gamma \in T_{\lambda,d_{\cH}}$. Expressing the polynomial  $G_{\tau,\gamma}(X)$ as a linear combination of monomials can be done in $\poly(n)$ time, for fixed $d_{\cH}$. 
 \end{lemma}   
 
 We now prove Lemma~\ref{lem:computing_entries_block}.
 
 %For $r \in [m^{\cH}]$, by using the fact that $C_{r}^{\cH} \in \cH^{\otimes n} \otimes \cH^{\otimes n}$ and $u_{\tau},u_{\gamma} \in (\cH^{\otimes n})^{*}$ via standard inner product (here we can consider $u_{\tau},u_{\gamma}$ as column vectors). We can write $u_{\tau}^{T}C_{r}^{\cH}u_{\gamma} = (u_{\tau} \otimes u_{\gamma})(C_{r}^{\cH})$. Set 
 
 %we can write $u_{\tau}^{T}C_{r}^{\cH}u_{\gamma} = (u_{\tau} \otimes u_{\gamma})(C_{r}^{\cH})$, here we are considering $u_{\tau},u_{\gamma}$ as column vectors (hence $u_{\tau},u_{\gamma} \in (\cH^{\otimes n})^{*}$ via standard inner product) and fact that $C_{r}^{\cH} \in \cH^{\otimes n} \otimes \cH^{\otimes n}$ for all $r \in [m(\cH)]$. Set

 \begin{proof}[Proof Lemma~\ref{lem:computing_entries_block}]
 	
 	For $r \in [m(\cH)]$, utilizing the fact that $C_{r}^{\cH} \in \cH^{\otimes n} \otimes \cH^{\otimes n}$ and $u_{\tau}, u_{\gamma} \in (\cH^{\otimes n})^{*}$ (via the standard inner product), we can write $u_{\tau}^{T}C_{r}^{\cH}u_{\gamma} = (u_{\tau} \otimes u_{\gamma})(C_{r}^{\cH})$. Set
 	\begin{align*}
 		g \coloneqq u_{\tau} \otimes u_{\gamma}  =  \sum_{\substack{\tau' \sim \tau \\ \gamma' \sim \gamma}}\sum_{c,c' \in C_{\lambda}} \mathrm{sgn}(cc') \bigotimes_{y \in Y(\lambda)} (F)_{\tau'c(y),\gamma'c'(y)} \, , 
 	\end{align*} 
 	where $F \in (W^{*})^{d_{\cH} \times d_{\cH}}$ with $(F)_{x,y} = a_{(x,y)}^{*}$. Then
 	
 	\begin{align*}
 		\sum_{r=1}^{m(\cH)}(u_{\tau} \otimes u_{\gamma})(C_{r}^{\cH}) \mu(O_{r}^{\cH}) &= \sum_{r=1}^{m(\cH)}g(C_{r}^{\cH})\mu(O_{r}^{\cH}) \\
 		&= \sum_{(p_1,\dots,p_n) \in ([d_{\cH}] \times [d_{\cH}])^n} g(a_{p_1} \otimes \dots \otimes a_{p_n})a_{p_1}^{*}\cdots a_{p_{n}}^{*} \\
 		&= \sum_{(p_1,\dots,p_n) \in ([d_{\cH}] \times [d_{\cH}])^n} g(a_{p_1} \otimes \dots \otimes a_{p_n})\overline{a_{p_1}^{*}\otimes \dots \otimes a_{p_n}^{*}} \\
 		&= \overline{g} = \sum_{\substack{\tau' \sim \tau \\ \gamma' \sim \gamma}} \sum_{c,c' \in C_{\lambda}} \mathrm{sgn}(cc') \prod_{y \in Y(\lambda)} (F)_{\tau'c(y),\gamma' c'(y)}  = G_{\tau,\gamma}(F) \, .
 	\end{align*}
 	
 	By the Lemma~\ref{lem:polynomial}, one computes the entry $\sum_{r=1}^{m(\cH)}z_{r}u_{\tau}^{T}C_{r}^{\cH}u_{\gamma}$ by replacing each monomial $\mu(O_{r}^{\cH})$ in $G_{\tau,\gamma}(F)$ with the variable $z_r$. Thus, given $z_1,\dots,z_{m(\cH)}$ we can compute the value $\sum_{r=1}^{m(\cH)}z_{r}u_{\tau}^{T}C_{r}^{\cH}u_{\gamma}$ in $\poly(n)$ time. Finally, from Eq. \eqref{eq:compute_bolck} we can determine the matrix 
 	$\left \llbracket \psi_{\cA,\cH,\cB} \left(\ketbra{i_\cA}{j_\cA} \otimes C_t^{\cH} \otimes \ketbra{i_\cB}{j_\cB} \right) \right \rrbracket_{\lambda}$ in $\poly(d_{\cA},d_{\cB},n)$ time for all $\lambda \in \Par(d_\cH,n)$ and $i_\cA,j_\cA \in [d_\cA]; i_\cB,j_\cB \in [d_\cB]; t \in [m(\cH)]$, which completes the proof.
 	%The proof can be found in Appendix~\ref{Apen:computing_entries}. 
 \end{proof}

\end{document}